\documentclass{article}
\usepackage{amsmath,amsfonts,amssymb,amsthm,graphicx}
 \usepackage{dsfont}     
\usepackage{enumerate} 
\usepackage{color}
\usepackage{natbib}
\usepackage{multirow}
\usepackage{comment}
\usepackage{indentfirst}
\usepackage{booktabs,subcaption,dcolumn}

\usepackage[inline, shortlabels]{enumitem}

\usepackage{algorithm}
\usepackage[noend]{algpseudocode}  
\usepackage[colorlinks=true,citecolor=blue,pdfpagemode=UseNone,pdfstartview=FitH]{hyperref}

\usepackage[margin=1.2in]{geometry}

\renewcommand{\d}{\,\mathrm{d}}

\newcommand{\p}{\mathbb{P}}

\newcommand{\E}{\mathbb{E}}    % expectation 
\newcommand{\R}{\mathbb{R}}    % real numbers
\newcommand{\N}{\mathbb{N}}    % natural numbers

\newcommand{\id}{\mathds{1}} 

\newcommand{\cD}{\mathcal{D}}

\usepackage{accents}
\newcommand{\ubar}[1]{\underaccent{\bar}{#1}}
\newcommand{\boldP}{(P_k)_{k \in \mathcal{K}}}
\newcommand{\boldE}{(E_k)_{k \in \mathcal{K}}}
\newcommand{\boldPstar}{(P^*_k)_{k \in \mathcal{K}}}
\newcommand{\boldSsq}{(S^2_k)_{k \in \mathcal{K}}}

\usepackage{setspace}

\theoremstyle{plain}
\newtheorem{theorem}{Theorem}[section]

\newtheorem{proposition}[theorem]{Proposition}

\theoremstyle{definition}
\newtheorem{definition}[theorem]{Definition}

\newtheorem{assumption}[theorem]{Assumption}

\theoremstyle{remark}
\newtheorem{remark}[theorem]{Remark}

 \renewcommand{\cite}{\citet}  % should be avoided

\title{E-values as unnormalized weights in multiple testing} 
\author{
Nikolaos Ignatiadis\thanks%
 {Department of Statistics, University of Chicago.
 E-mail: \href{mailto:ignat@uchicago.edu}{ignat@uchicago.edu}.}
 \and
 Ruodu Wang\thanks%
 {Department of Statistics and Actuarial Science,
 University of Waterloo.
 E-mail: \href{mailto:wang@uwaterloo.ca}{wang@uwaterloo.ca}.}
\and  
Aaditya Ramdas\thanks%
    {Departments of Statistics and Machine Learning,
    Carnegie Mellon University.
    E-mail: \href{mailto:aramdas@cmu.edu}{aramdas@cmu.edu}.}
}

\begin{document}

\maketitle

\begin{abstract}
We study how to combine p-values and e-values, and design multiple testing procedures where both p-values and e-values are available for every hypothesis. Our results provide a new perspective on multiple testing with data-driven weights: while standard weighted multiple testing methods require the weights to deterministically add up to the number of hypotheses being tested, we show that this normalization is not required when the weights are e-values that are independent of the p-values. Such e-values can be obtained in the meta-analysis setting wherein a primary dataset is used to compute p-values, and an independent secondary dataset is used to compute e-values. Going beyond meta-analysis, we showcase settings wherein independent e-values and p-values can be constructed on a single dataset itself. Our procedures can result in a substantial increase in power, especially if the non-null hypotheses have e-values much larger than one. 
~
\bigskip

\end{abstract}

\section{Introduction}
\label{sec:1}

The p-value is perhaps the most commonly used inferential device in statistical practice. Traditional procedures for multiple testing, such as the procedure of \cite{BH95} for controlling the false discovery rate, begin with a list of p-values as the input. The e-value is an alternative inferential tool that  encompasses betting scores, likelihood ratios, and stopped supermartingales, e.g., \cite{S20, VW20, GDK20}, and \cite{HRMS20,HRMS21}.  For example, the ``universal inference'' e-value has gained popularity, and has recently led to the first known valid tests for many composite null hypotheses, such as testing mixtures, e.g., testing if the data comes from a mixture of Gaussians~\citep{WRB20}, or testing for shape constraints, e.g., testing if the data distribution is log-concave~\citep{dunn2021universal}. Lists of e-values can also serve as the input to multiple testing procedures~\citep{WR21,XWR21}.

In this paper, we design testing procedures 
for situations in which we have both a p-value and an e-value for each hypothesis. A first motivation for our proposed methods is the \emph{meta-analysis} setting wherein we collect data from two distinct sources. 
Our contributions to meta-analysis acknowledge our anticipation that e-values will increasingly find adoption in applications without displacing p-values.
Thus it is natural to develop procedures that can optimally combine the information available in an e-value and a p-value. What's more, we argue that our proposed meta-analysis methods are useful even when the analyst could in principle compute two separate p-values, one on each distinct dataset. Our methods provide an alternative to other existing meta-analysis methods~\citep{heard2018choosing} that can be particularly powerful when one dataset (the primary dataset) is more informative than the secondary dataset. (We provide theoretical and empirical justification in Sections \ref{subsec:stylized_two_sample_main} and~\ref{sec:sim_rnaseq_microarray}.)

As a second contribution, our methods provide a new perspective on multiple testing with \emph{data-driven hypothesis weights}.
Weighted multiple testing procedures provide a flexible and convenient way of differentially prioritizing hypotheses by assigning a weight to each hypothesis and prioritizing hypotheses with large weights~\citep{BH97, GRW06, BR08, RBWJ19}. If the weight assignment is informative and correctly prioritizes alternatives, then
weighted multiple testing procedures can lead to substantial power gains compared to unweighted procedures.
Weighting methods have traditionally come with two requirements: first, the weights need to be deterministic, that is, they should not depend on the data used to compute the p-values, and second, they need to average to $1$. Intuitively, the first requirement implies that the weights can only be a priori ``guesses,'' and the second requirement enforces a constrained size budget to be split across hypotheses. A nascent literature including e.g.,~\citet{westfall2004weighted, finos2007fdr, roeder2009genome, ignatiadis2016data, durand2019adaptive, ignatiadis2021covariate}, has dispensed with the first requirement: it is possible to construct data-driven weights and p-values based on the same dataset. 
In this paper, we demonstrate (for the first time, to our knowledge) that it is also simultaneously possible to dispense with the fixed weight budget requirement.

The key insight for our contributions to both meta-analysis and data-driven hypothesis weighting is the following: independent e-values can be directly used as weights for p-values in all standard multiple testing procedures, without needing to normalize them in any way. This can lead to huge increases in power relative to standard weighted procedures.

\section{Multiple testing background}

\subsection{Terminology and notation} \label{sec:terminology}
We first describe the basic setting. 
Let $H_1,\ldots,H_K$ be $K$ hypotheses, and write $\mathcal K=\{1,\ldots,K\}$.
Let the true (unknown) data-generating  probability measure be denoted by $\p$.
For each $k\in \mathcal K$, 
it is useful to think of hypothesis $H_k$ as implicitly defining a set of joint probability measures, and $H_k$ is called a true null hypothesis if $\p \in H_k$. 
A \emph{p-value} $P$ for a hypothesis $H$ is a  random variable  that satisfies $Q(P\le t)\le t$ for all $t \in [0,1]$ and all $Q\in H$. 
In other words, a p-value is stochastically larger than $\mathrm {U}(0,1)$. An \emph{e-value} $E$ for a hypothesis $H$ is a $[0,\infty]$-valued random variable satisfying $\E^Q(E)\le1$ for all $Q\in H$. 
Let $\mathcal N\subseteq \mathcal K$ be the (unknown to the decision maker) index set of true null hypotheses, $K_0 := | \mathcal{N}|$ the number of true null hypotheses, and $\pi_0 := K_0/K$ the proportion of true null hypotheses.
 
Two settings of testing multiple hypotheses were considered by \cite{WR21}. In the first setting, for each $k\in \mathcal K$,   $P_k$ is a   p-value for $H_k$. In the second setting, for each $k\in \mathcal K$,   $E_k$ is an e-value for $H_k$. In this paper 
we will consider the setting where both $P_k$ and $E_k$ are available for each $H_k$. 
Since we are testing  whether $\p\in H_k$ for each $k$, we will only use the following (obvious) condition:
if $k\in \mathcal N$, then $\p(P_k \leq t) \leq t$ for all $t \in [0,1]$ and $\mathbb{E}^{\p}(E_k)\leq 1$.
% $P_k$ is a p-value for $\{\p\}$, and $E_k$ is an e-value for $\{\p\}$. 
There are no restrictions on $P_k$ and $E_k$ if $k\not \in \mathcal N$.
We will omit $\p$ in the statements (by simply calling them p-values and e-values) and the expectations. The terms p-values/e-values refer to both the random variables and their realized values (these should be clear from the context).

Now let $\cD$ be a testing procedure, that is, a Borel mapping that produces a subset of $\mathcal K$
representing the indices of rejected hypotheses based on p-values (we write p-$\cD$ to denote a procedure $\cD$ that is  based only on p-values), e-values (e-$\cD$), or a combination of both as the input.  
The rejected hypotheses by  $\cD$ are called discoveries. 
We write $F_{\cD}:=|\cD \cap \mathcal N|$ as the number of true null hypotheses that are rejected (i.e., false discoveries), and $R_{\cD}:=|\cD|$ as the total number of discoveries. We are interested in controlling generalized type-I errors that are defined as expectations of the form $\E\{G(F_{\cD}, R_{\cD})\}$, where $G: \mathbb N_{\geq 0} \times \mathbb N_{\geq 0} \to \mathbb R_{\geq 0}$ is a fixed mapping.
 
One choice of particular interest is the choice $G(f, r) = f/r$, with the convention $0/0=0$. Then $G(F_{\cD}, R_{\cD}) =  F_{\cD}/R_{\cD}$ is 
 called the false discovery proportion, which is the ratio of the number of false discoveries to that of all claimed discoveries. 
 \cite{BH95} proposed to control the false discovery rate, which is the expected value of the false discovery proportion,
that is,
$  
 \mathrm{FDR}_{\cD}:=\E ( {F_{\cD}}/{R_{\cD} } ).
$ 
Further important generalized type-I errors are given by the choices $G(f, r)= f$ and $G(f,r) = \id( f \geq 1)$.  These yield the per-family error rate  of a procedure $\cD$ which is defined as
$\mathrm{PFER}_{\cD}:= \E(F_{\cD})$, as well as the family-wise error rate, defined as 
 $\mathrm{FWER}_{\cD}:=\p(F_\cD \ge 1)$. The family-wise error rate is particularly relevant for testing the global null, and is identical to the false discovery rate if all hypotheses are true nulls.

 We next turn to discuss the dependence structure among p-values. A common, albeit strong assumption that appears in the literature, e.g., in~\citet{liang2012adaptive}, is the following:
\begin{definition}[P-Independence]
\label{definition:pindependence}
A vector $\boldP$ of p-values satisfies the p-independence property if:
\begin{enumerate*}[(i)]
    \item the null p-values $(P_k)_{k \in \mathcal{N}}$ are mutually independent, and
    \item the null p-values $(P_k)_{k \in \mathcal{N}}$ are independent of the non-null p-values $(P_k)_{k \notin \mathcal{N}}$.
\end{enumerate*}
\end{definition}
To relax the above assumption, we rely on the notion of positive regression dependence on a subset in \citet[Section 4]{finner2009false} and \citet{BR17} which is slightly weaker than the original one used in \cite{BY01}.
A set $A\subseteq \R^K$ is said to be \emph{increasing}
if $x\in A$ implies $y\in A$ for all $y\ge x$.   The term ``increasing" is in the non-strict sense,
and  inequalities should be
interpreted component-wise when applied to vectors.
\begin{definition}[Positive regression dependence on a subset]
\label{definition:prds}
A vector $\boldP$ of p-values satisfies positive regression dependence on a subset if for any null index $k\in \mathcal N$ and increasing set $A  \subseteq  \R^K$, the
function $x\mapsto \p\{(P_{\ell})_{\ell \in \mathcal{K}} \in A\mid P_k\le x\}$ is increasing on $ [0,1]$.  
\end{definition}
A caveat of Definition~\ref{definition:prds} is that it enforces certain positive dependence between the nulls and non-nulls. To address this concern, \citet{su2018fdrlinking} proposed the following more general notion of dependence.

\begin{definition}[Positive regression dependence within nulls]
\label{definition:prdn}
A vector $\boldP$ of p-values satisfies positive regression dependence within nulls if the subvector of null p-values,  $(P_k)_{k \in \mathcal{N}}$, is positive regression dependent on a subset.
\end{definition}

\subsection{Unweighted and weighted multiple testing procedures}\label{sec:2}
We now describe a few canonical procedures that control the generalized type-I errors  introduced above. We start by describing the p-BH and e-BH procedures. These procedures use p-values, respectively e-values, and seek to control the false discovery rate at the target level $\alpha$.
  
\begin{definition}[p-BH procedure (\citealp{BH95})]
\label{definition:pBH}
For $k\in \mathcal K$, let $P_{(k)}$ be the $k$-th order statistic of the p-values $P_1,\ldots,P_K$, from the smallest to the largest. The p-BH procedure rejects all hypotheses with the smallest $k_p^*$ p-values,
where 
\begin{equation}
\label{eq:p-k}
k_p^*:=\max\left\{k\in \mathcal K: \frac{K P_{(k)}}{k} \le \alpha\right\},
\end{equation}
with the convention $\max(\varnothing)=0$.
\end{definition}

\begin{definition}[e-BH procedure (\citealp{WR21})]
\label{definition:eBH}
For $k\in \mathcal K$, let $E_{[k]}$ be the $k$-th order statistic of the e-values $E_1,\ldots,E_K$, from the largest to the smallest.  The e-BH procedure rejects all hypotheses with the largest $k_e^*$ e-values, 
where 
\begin{equation} 
\label{eq:e-k-intro} 
k_e^*:=\max\left\{k\in \mathcal K: \frac{k E_{[k]}}{K} \ge \frac{1}{\alpha}\right\}.\end{equation}   
\end{definition}
\noindent An equivalent way to describe the e-BH procedure is to apply the p-BH procedure to $(E_1^{-1},\ldots,E_K^{-1})$. 

The p-BH procedure at level $\alpha $ has false discovery rate at most
\begin{enumerate*}[(i)]
\item $\pi_0\alpha$ when the p-values satisfy p-independence or positive regression dependence on a subset \citep{BH95,BY01},
\item $\pi_0 \alpha \log\{e/(\pi_0 \alpha)\}$ when the p-values satisfy positive regression dependence within nulls ~\citep{su2018fdrlinking}, and
\item  $\ell _K \pi_0 \alpha$, where 
$
\ell_K:=\sum_{k=1}^K  k^{-1} \approx \log K$, under arbitrary dependence~\citep{BY01}.
\end{enumerate*}
As for the e-BH procedure, 
\cite{WR21} showed a surprising property that 
the base e-BH procedure controls the false discovery rate at $\alpha$ even under unknown  {arbitrary dependence} between the e-values.

A procedure closely related to p-BH is the p-Simes procedure. This is not a multiple testing procedure per se, but instead, it is a test of the global null hypothesis $H:= \bigcap_{k=1}^K H_k$.
\begin{definition}[p-Simes procedure (\citealp{simes1986improved})]
\label{definition:pSimes}
The p-Simes procedure rejects the global null $ \bigcap_{k=1}^K H_k$ when the p-BH procedure applied to $\boldP$ makes at least one discovery.
\end{definition}
The p-Simes procedure has type-I error at most $\alpha$ when the p-values are positive regression dependent within nulls. 

We next present the p-Bonferroni procedure to control the per-family error rate and the family-wise error rate.
\begin{definition}[p-Bonferroni procedure (\citealp{bonferroni1935calcolo})]
\label{definition:pBonf}
Let $P_1,\dotsc,P_K$ be the p-values.
The p-Bonferroni procedure rejects all hypotheses with $P_k \leq \alpha/K$.
\end{definition}
The p-Bonferroni procedure controls the per-family error rate and the family-wise error rate at level $\alpha$ under arbitrary p-value dependence. The following procedure (p-Hochberg) controls the family-wise error rate under a stronger dependence assumption, namely, positive regression dependence within nulls, and
is more powerful than p-Bonferroni.
\begin{definition}[p-Hochberg procedure (\citealp{hochberg1988sharper})]
\label{definition:pHochberg}
For $k\in \mathcal K$, let $P_{(k)}$ be the $k$-th order statistic of the p-values $P_1,\ldots,P_K$, from the smallest to the largest. The p-Hochberg procedure rejects all hypotheses with the smallest $k_{h}^*$ p-values,
where 
\begin{equation*}
k_{h}^*:=\max\left\{k\in \mathcal K: P_{(k)} \le \frac{\alpha}{K-k+1}\right\}.
\end{equation*}
\end{definition}
In Supplement~\ref{sec:additional_procedures} we also discuss the procedures of~\citet{holm1979simple} and~\citet{hommel1988stagewise}.

Many p-value based multiple testing procedures may be applied alongside a vector of weights. Two examples are weighted p-BH and weighted p-Bonferroni~\citep{GRW06}.
\begin{definition}[Weighted p-BH and weighted p-Bonferroni procedures]
\label{definition:wBH}
Let $P_1,\dotsc,P_K$ be the p-values and let $(w_1,\ldots,w_K)\in [0,\infty)^K$ be a pre-specified vector of weights. The weighted p-BH procedure (resp. p-Bonferroni procedure) is obtained by applying the p-BH (resp. p-Bonferroni procedure) to $(P_1/w_1,\ldots,P_K/w_K)$.
\end{definition} 
For generalized type-I error control, classical thinking imposes the fixed weight budget requirement that the weights are normalized and average to $1$, that is, $\sum_{k=1}^K w_k = K$. In that case, the weighted p-BH procedure controls the false discovery rate when the p-values are positive regression dependent on a subset~\citep{BR08, RBWJ19}, and the weighted p-Bonferroni procedure controls the per-family error rate and the family-wise error rate under arbitrary dependence of the p-values~\citep{GRW06}.
Later, we will see if the weights are obtained from e-values independent of the p-values, then normalization is not needed, and this can improve power substantially.

\section{Combining a p-value and an e-value}\label{sec:3}
 
\subsection{Admissible p-value/e-value combiners}

One of the main objectives of the paper 
is to design and understand procedures when both p-values and e-values are available. 
For this purpose, we first look at the 
  single-hypothesis setting, in which case we drop the subscripts and use $P$  for a p-value and $E$ for an e-value.

We briefly review calibration between a p-value and an e-value as developed previously by \cite{SSVV11} and \citet[Chapter 11.5]{shafer2019game}, amongst other sources.  
Denote by $\overline{\R}_+=[0,\infty]$. 
 First, an e-value $E$ can be converted to a p-value $P =  (1/E)\wedge 1 $ (its validity follows from Markov's inequality). Further, the function $f:e \mapsto (1/e)\wedge 1$ is the unique admissible e/p calibrator \citep[Proposition 2.2]{VW20}.

A p-value $P$ can also be converted to an e-value, but there are many admissible choices. One example is to set $E = P^{-1/2} - 1$.   
 More generally, we   speak of p/e calibrators. 
Small p-values correspond to large e-values, which represent stronger evidence against a null hypothesis.
A p/e calibrator is a decreasing function $h:[0,1]\to \overline{\R}_+$ satisfying $\int_0^1 h(u)\d u \leq 1$.  Then $h(P)$ is an e-value for any p-value $P$.
 \citet[Proposition 2.1]{VW20} show that the set  $\mathcal C^{\rm p/e}$ of all admissible p/e calibrators is 
 $$
\mathcal C^{\rm p/e} =\left\{ h: [0,1]\to \overline{\R}_+  \mbox{ decreasing \& upper semicontinuous } \mid  h(0)=\infty, \int h(u)\d u = 1\right\}.
$$ 
In the above statements, admissibility of a calibrator (or a combiner below) means that it cannot be improved strictly, where improvement means obtaining a larger e-value or a smaller p-value.

Combining several p-values or e-values to form a new p-value or e-value is the main topic of \cite{VW19, VW20} and \cite{VWW22}. For the objective of this paper, we need to combine a p-value $P$ and an e-value   $E$, first in a single-hypothesis testing problem. We consider four cases.
\begin{enumerate*}[(i)]
    \item If $P$ and $E$ are independent, how should we combine them to form an e-value?
    \item  If $P$ and $E$ are independent, how should we combine them to form a p-value?
    \item  If $P$ and $E$ are arbitrarily dependent, how should we combine them to form an e-value?
    \item  If $P$ and $E$ are arbitrarily dependent, how should we combine them to form a p-value?
\end{enumerate*}

We use the following terminology, similar to  \cite{VW20}. 
A function $f: [0,1]\times \overline{\R}_+ \to  \overline{\R}_+$
is called an i-pe/e combiner 
if  $f(P,E)$ is an e-value for any independent p-value $P$ and e-value $E$,
and $(p,e)\mapsto f(p,e)$ is decreasing in $p$ and increasing in $e$.
Similarly, we define i-pe/p, pe/p, and pe/e combiners, where i indicates independence, and p and e are self-explanatory.
If the output is a p-value,
the combiner is increasing in $p$ and decreasing in $e$.

We provide four natural answers to the above four questions, some relying on   an admissible calibrator $h \in \mathcal C^{\rm p/e}$.
\begin{enumerate*}[(i)]
    \item Return $h(P)  E$ by using the function $\Pi_h(p,e):=h(p)e$. The convention here is $0\times \infty=\infty$. 
    \item  Return $P/E$, capped at $1$, by  using the function $Q(p,e):=(p/e)\wedge 1$.
    \item  Return $\lambda h(P)+(1-\lambda) E$  by   using the function $M^\lambda_h(p,e):= \lambda h(p)+(1-\lambda) e $ for some $\lambda \in (0,1)$.
    \item  Return $2\min(P,1/E)$, capped at $1$,  by  using the function $B(p,e):= \{ 2 (  p\wedge  e^{-1}) \}\wedge 1$.
\end{enumerate*}

The notation chosen for these functions is due to the initials of (i) product (but we avoid $P$ which is reserved for p-values); (ii) quotient; (iii) mean; (iv) Bonferroni correction. % on $p$ and $e^{-1}$. 

$\Pi_h$ and $M^\lambda_h$ depend on $h$ whereas $Q$ and $B$ do not. 
For the function $M^{\lambda}_h$, it may be convenient to choose $\lambda=1/2$, so that $M^{\lambda}_h(P,E)$ is the arithmetic average of two e-values $h(P)$ and $E$. As shown by \citet[Proposition 3.1]{VW20}, the arithmetic average essentially dominates, in a natural sense, all symmetric e-merging function. 
In our context, $\lambda=1/2$ has no special role, since the positions of $h(P)$ and $E$ are not symmetric.

 \begin{theorem} \label{th:combiners}
For $h\in \mathcal C^{\rm p/e}$ and $\lambda \in(0,1)$,
 $\Pi_h$ is an admissible i-pe/e combiner, 
 $Q$ is an admissible i-pe/p combiner, 
 $M^\lambda_h$ is an admissible pe/e combiner,   
 $B$ is an admissible pe/p combiner.
 \end{theorem} 
The proof can be found in Supplement~\ref{sec:proof-admissible}.
For the remainder of the paper, we pay particular attention to the i-pe/p combiner $Q$ that forms a p-value based on independent $P$ and $E$. We use the term $Q$-combiner to refer to both the mapping $(P,E) \mapsto Q(P,E)=(P/E)\land 1$ as well as the resulting p-value $Q(P,E)$.
The $Q$-combiner typically leads to more powerful procedures compared to the other combiners and provides the foundation for our insight that e-values can act as unnormalized weights in multiple testing (see next sections).
The $\Pi_h$ i-pe/e combiner is also of interest, and we develop  results for $\Pi_h$ in the context of multiple testing in Supplement~\ref{sec:mtp_with_pi_h}.

\begin{remark}
One consequence of Theorem~\ref{th:combiners} is as follows. Consider an e-value $E$ and generate an independent uniform variable $U \sim \mathrm{U}(0,1)$. Then, $P':=Q(U, E)$ is a valid p-value that satisfies $\mathbb P\{Q(U, E) \leq f(E)\}=1$ and $\mathbb P\{Q(U, E) < f(E)\} > 0$, where $f:e \mapsto (1/e)\wedge 1$ is the unique admissible e/p calibrator. Hence, $f$ is dominated by a randomized e/p calibrator.
Although $Q(U,E)$ may not be  practical in general due to external randomization, it becomes practical when applied as $Q(P,E)$ to a p-value $P$ (computed from data) independent of $E$.
\end{remark}

\subsection{\texorpdfstring{$Q$}{Q}-combiner as a general-purpose method for meta-analysis from two studies}

As mentioned above, for the remainder of the paper we consider procedures that build on 
the $Q$-combiner $(P,E) \mapsto Q(P,E)=(P/E)\land 1$. To start, we argue that the $Q$-combiner is a useful general-purpose method for meta-analysis from two independent datasets. The $Q$-combiner is immediately applicable when the researcher summarizes the first dataset as a single p-value, and the second dataset as a single e-value. 
Such a situation could occur when the second dataset is collected in such a way, e.g., with optional stopping and continuation, that inference is more natural with e-values; see~\citet{ramdas2022gametheoretic} for a survey of e-values and the inferential problems they solve. It could also be the case that one dataset comprises of a large sample size, allowing for asymptotic approximations to compute p-values, while the second dataset is smaller and may require finite-sample inference methods, e.g., universal inference~\citep{WRB20}, that lead to e-values.

Our claim, however, is stronger: the $Q$-combiner is also useful when the above data constraints are not in place and the researcher can in principle compute both a p-value $P'$ and an e-value $E$ on the second dataset, both of which are independent of the p-value $P$ computed on the first dataset. In that case, the researcher could apply a p-value combination method based on $P$ and $P'$, e.g., Fisher's combination  $P_{\rm F}:=1- \chi_4\{-2\log ( P P')\}$, 
where $\chi_4$ is the chi-square distribution with $4$ degrees of freedom. However, the researcher may still prefer to proceed with the $Q$-combiner $Q(P,E)$. We suggest the following rule of thumb.

\emph{The Fisher combination is preferable to the $Q$-combiner under dataset exchangeability:} Suppose that the analyst considers the two datasets as a priori exchangeable. In that case, it may be undesirable to use an asymmetric combination rule such as $Q(P,E)$, and Fisher's combination $P_{\rm F}$ is preferable on conceptual grounds. If the two datasets are also exchangeable in terms of their statistical properties (i.e., they have similar power), then $P_{\rm F}$ will typically have higher power than $Q(P,E)$. 

\emph{The $Q$-combiner is preferable to the Fisher combination for imbalanced datasets:} When one dataset (the ``primary'' dataset) is substantially more well-powered (larger anticipated signal or sample size) than the secondary dataset, and the investigator knows which dataset is more well-powered, then the $Q$-combiner can often outperform Fisher's combination test in terms of power. A proviso is that the p-value is computed on the primary (more well-powered) dataset and the e-value on the secondary dataset.

In the next section, we provide theoretical and numerical evidence for the rule of thumb put forth in the preceding paragraph in a stylized example. We also provide further numerical evidence in the simulations of Section~\ref{sec:sim_rnaseq_microarray}.

\subsection{A stylized example: using two samples for the one-sided z-test via the \texorpdfstring{$Q$}{Q}-combiner}
\label{subsec:stylized_two_sample_main}
As a stylized example, suppose we have access to two independent samples of iid data points, $X=(X_1,\dots,X_m)$  and $ Y=(Y_1,\dots,Y_n)$, both from a distribution $\p$, where  $n\ge m \ge 1$. 
We seek to test $H_0: \p = \mathrm{N}(0,1)$ 
against $H_1: \p=\mathrm{N}(  \delta,1)$,
where $\delta >  0$ is known.
The optimal p-value based on $X$ is $P_{X} := 1-\Phi(T_{X})$,
 where $\Phi$  is the standard normal distribution function and $T_{ X}:=\sum_{i=1}^m X_i/\surd{m}$. Analogously we may compute p-values $P_{Y}$ based on $ {Y}$, as well as $P_{Z}$, where $Z=(X, Y)$ is the full dataset. The optimal e-value $E_{X}$ based on $X$ is the likelihood ratio of $\mathrm{N}(\delta,1)^m$ over $\mathrm{N}(0,1)^m$. 
 
By the Neyman--Pearson lemma, the p-value $P_{Z}$ leads to the most powerful test. We seek to compare $P_{Z}$ against the $Q$-combiner $P_{\rm E} := Q(P_{Y},E_{X})$  by considering the hypothesis tests that reject $H_0$ when $P_{Z} \leq \alpha$, or when $P_{\rm E} \leq \alpha$, for $\alpha > 0$.
We assume $m=\lfloor \theta^2 n \rfloor$ for some $\theta \in (0,1]$ which measures the relative size of the two datasets.

In Supplement~\ref{sec:two_sample_stylized}, we derive Pitman's asymptotic relative efficiency (\citealp[Section 14.3]{V98}; \citealp[Section 22.1]{D08}) between the two methods, 
which is the asymptotic ratio of the required sample size from $P_{Z}$ to reach a fixed power, to that from $P_{\rm E}$, as $\delta \downarrow 0
$.
%between these two tests as $m+n \to \infty$ such that $m/n \to \theta \in (0,1]$ and $\delta \to 0$:
We prove that the asymptotic relative efficiency converges to $1$ in two different settings: as $\alpha \downarrow 0$, that is, when the type-I error is very stringent, and as $\theta \downarrow 0$, that is, when $Y$ is substantially more well-powered than $X$. Our results can also be used to numerically compute the asymptotic relative efficiency for any choice of $\theta, \alpha$, and desired power, e.g., the asymptotic relative efficiency is (up to numerical rounding) equal to $0.989$ when $\theta = 0.5$, $\alpha=0.05$, and we seek a power of $50\%$.

We also conduct a small simulation study comparing 
\begin{enumerate*}[(i)]
\item $P_{Z}$,
\item $P_{\rm E}$, and,
\item the Fisher p-value $P_{\rm F}:=1- \chi_4\{-2\log ( P_{ Y} P_{ X})\}$.
\end{enumerate*}
Simulation results are reported in Fig.~\ref{fig:0}  based on the average of 10,000 runs. We take $\alpha=0.05$, $m+n=100$ and let the ratio $m/n$ and $\delta >0$ vary. We observe the following: if $m/n$ is small (first two panels), meaning that $X$ is less informative than $Y$,  then the $Q$-combiner has almost the same power as the full likelihood ratio method, and both outperform the Fisher method. When $m=n$ (third panel), the Fisher test has more power than the $Q$-combiner, and both have (slightly) less power than the likelihood ratio test.

\begin{figure}
\centering
\includegraphics[width=\linewidth]{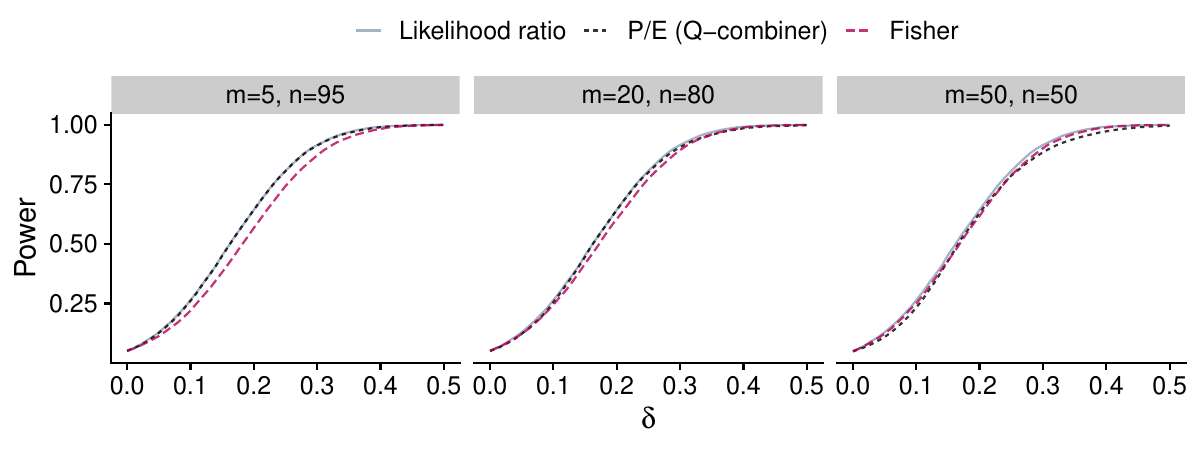}
\caption{Simulation study for a meta-analysis combining two samples: We compare the likelihood ratio test, the $Q$-combiner, and the Fisher combination test, plotting power against signal strength $\delta$. The panels correspond to different choices of the two sample sizes $m$ and $n$. The $Q$-combiner is visibly more powerful on the left (matching the likelihood ratio), and Fisher's combination is marginally better on the right.}
\label{fig:0}
\end{figure}

 \section{Using e-values as weights in multiple testing with p-values}\label{sec:4}

\subsection{General remarks}
In this section, we 
move back to   multiple testing by considering the setting where each hypothesis is associated with a p-value and an e-value. 
The generalized type-I error of testing procedures depends on the dependence amongst the e-values, the dependence amongst the p-values, and the dependence between the p-values and e-values. Throughout this section we make the following assumption:
\begin{assumption}
\label{assumption:coordinatewise}
$P_k$ is independent of $E_k$ for all $k \in \mathcal{N}$.
\end{assumption}

Given any procedure p-$\cD$ that is based on p-values, we can extend it to an e-weighted procedure that we call ep-$\mathcal{D}$ and which generalizes the concept of weighting for multiple testing. 
\begin{definition}[e-weighted p-value procedure (ep-$\mathcal{D}$)]
\label{definition:eweightedpvalue}
Let p-$\cD$ be a multiple testing procedure based on p-values. Given p-values $\boldP$ and e-values $\boldE$, we define the e-weighted p-value procedure ep-$\cD$ which proceeds as follows: for $k\in \mathcal K$, compute the $Q$-combiner $P^*_k := Q(P_k,E_k) = (P_k/E_k)\land 1$,
and then supply $\boldPstar$ to p-$\mathcal{D}$. 
\end{definition}
Concretely, we define the
\begin{enumerate*}[(i)]
\item ep-BH,
\item ep-Simes,
\item ep-Bonferroni, resp.
\item ep-Hochberg procedure
\end{enumerate*}
by plugging in the 
\begin{enumerate*}[(i)]
\item p-BH (Definition~\ref{definition:pBH}),
\item p-Simes (Definition~\ref{definition:pSimes}),
\item p-Bonferroni (Definition~\ref{definition:pBonf}), resp.
\item p-Hochberg (Definition~\ref{definition:pHochberg}) 
\end{enumerate*}
procedure into Definition~\ref{definition:eweightedpvalue}.

In view of Definition~\ref{definition:wBH}, ep-BH and ep-Bonferroni may be interpreted in two ways: 
\begin{enumerate*}[(i)]
    \item they are p-value based procedures applied to the p-value vector $\boldPstar$, and
    \item they are weighted p-value based procedures with p-value vector $\boldP$ and weight vector $\boldE$.
\end{enumerate*}
Both perspectives are useful in deriving guarantees on the control of generalized type-I error rates.

\subsection{E-weighted p-value procedures as p-value procedures}

We first present a general result under Assumption~\ref{assumption:coordinatewise}. Recall the generalized type-1 error mapping $G$ from Section~\ref{sec:terminology} whose expectation captures error metrics like the false discovery rate, per-family error rate, and the family-wise error rate, amongst others.

\begin{theorem}
\label{th:epD_arbitrary_dependence}
    Let p-$\mathcal{D}$ be a p-value procedure such that  $\E\{G(F_{\text{p-}\cD}, R_{\text{p-}\cD})\} \leq \alpha'$ for any p-value vector that may be arbitrarily dependent, where $\alpha' > 0$. Suppose 
    further that Assumption~\ref{assumption:coordinatewise} holds. Then, the ep-$\mathcal{D}$ procedure applied to $\boldP, \boldE$ also satisfies
    $\E\{G(F_{\text{ep-}\cD}, R_{\text{ep-}\cD})\} \leq \alpha'.$
In particular,
    \begin{enumerate*}[(i)]
    \item the ep-BH procedure has false discovery rate at most $\ell _K \pi_0 \alpha$, where $\ell_K:=\sum_{k=1}^K  k^{-1}$, and
    \item the ep-Bonferroni procedure has per-family error rate and family-wise error rate at most $\alpha$.
    \end{enumerate*}
\end{theorem}

\begin{proof}
By Theorem~\ref{th:combiners}, $P_k^*$ is a valid p-value under Assumption~\ref{assumption:coordinatewise} for all $k \in \mathcal{N}$. Hence $\boldPstar$ is a valid p-value vector (that may be arbitrarily dependent).
\end{proof}

The above argument can easily be generalized. For example, if Assumption~\ref{assumption:coordinatewise} holds, and $\boldPstar$ satisfies positive regression dependence on a subset, then ep-BH has false discovery rate at most $\pi_0\alpha$. Analogously, if Assumption~\ref{assumption:coordinatewise} holds and $\boldPstar$ is positive regression dependent within nulls, then ep-Hochberg has family-wise error rate at most $\alpha$ and so forth.
The assumption that $\boldPstar$ is positive regression dependent on a subset (or within nulls), however, may be difficult to interpret. Thus we prefer to directly impose assumptions on $\boldP$,  $\boldE$, as well as the cross-dependence between $\boldP$ and $\boldE$.

We provide an example of the general approach
by considering the following independence assumption. In the next subsection, we provide more elaborate results by turning to the perspective that ep-$\mathcal{D}$ procedures are weighted procedures.

\begin{assumption}
\label{assumption:pe_is_indep}
\begin{enumerate*}[(i)]
\item The null (p-value, e-value) pairs $\{(P_k, E_k)\}_{k \in \mathcal{N}}$ are mutually independent, and
\item $\{(P_k, E_k)\}_{k \in \mathcal{N}}$ is independent of $\{(P_k, E_k)\}_{k \notin \mathcal{N}}$.
\end{enumerate*}
\end{assumption}

\begin{theorem}
\label{th:indep}
    Let p-$\mathcal{D}$ be a p-value procedure such that  $\E\{G(F_{\text{p-}\cD}, R_{\text{p-}\cD})\} \leq \alpha'$ for any p-value vector that satisfies p-independence (Definition~\ref{definition:pindependence}), where $\alpha' > 0$. Suppose 
    further that Assumptions~\ref{assumption:coordinatewise} and ~\ref{assumption:pe_is_indep} hold. Then, the ep-$\mathcal{D}$ procedure applied to $\boldP, \boldE$ also satisfies
    $\E\{G(F_{\text{ep-}\cD}, R_{\text{ep-}\cD})\} \leq \alpha'.$
In particular,
    \begin{enumerate*}[(i)]
    \item the ep-BH procedure has false discovery rate at most $\pi_0\alpha$, 
    \item the ep-Hochberg procedure has family-wise error rate at most $\alpha$.
    \end{enumerate*}
\end{theorem}

\begin{proof}
Assumptions~\ref{assumption:coordinatewise} and ~\ref{assumption:pe_is_indep} imply that $\boldPstar$ satisfies p-independence (Definition~\ref{definition:pindependence}). 
\end{proof}

\subsection{E-weighted p-value procedures as weighted p-value procedures}

We now turn to the second perspective of e-weighted procedures: we interpret the e-values as weights for the p-values. Intuitively, if $E_k>1$, then there is evidence against $H_k$ being a null, and we have $P_k/E_k<P_k$ (assuming $P_k\ne 0$), that is, the weight strengthens the signal of $P_k$. Conversely, if $E_k<1$, then there is no evidence against $H_k$ being a null, and  we have  $P_k/E_k>P_k$. 
The above interpretation of e-values as weights is quite natural, and the perspective is useful in deriving guarantees for, e.g., ep-BH, that may be challenging to prove otherwise: below we prove that ep-BH controls the false discovery rate under the assumption that $\boldP$ is positive regression dependent on a subset along with the following strengthening of Assumption~\ref{assumption:coordinatewise}.
\begin{assumption}
\label{assumption:full_indep}
$\boldP$ is independent of $\boldE$.
\end{assumption}
Positive regression dependence on a subset of $\boldP$, together with Assumption~\ref{assumption:full_indep}, does not imply positive regression dependence on a subset of $\boldPstar$,
and hence some arguments are needed to establish control of the false discovery rate by ep-BH. The following result integrates over the randomness in the weights (e-values). In contrast, weighted p-BH with normalized weights controls the false discovery rate conditionally on all the weights.

\begin{theorem}\label{th:ep-BH}
Suppose that Assumption~\ref{assumption:full_indep} holds and that $\boldP$ is positive regression dependent on a subset (Definition~\ref{definition:prds}). Then, the ep-BH procedure has false discovery rate at most $\pi_0 \alpha$. 
\end{theorem}

\begin{proof}
Let ep-$\cD$ be the ep-BH procedure at level $\alpha$. 
Since $\boldE$ is independent of $\boldP$, 
conditional on $\boldE$, the ep-BH procedure becomes a weighted p-BH procedure with weight vector 
$\boldE$ applied to the p-values $\boldP$ that are positive regression dependent on a subset. Using well-known existing results on the false discovery rate of the weighted p-BH procedure (e.g., \citealp[Theorem 1]{RBWJ19}), 
we get  
 $$
\E\left\{\frac{F_{\text{ep-}\cD}}{R_{\text{ep-}\cD} }\;\Big |\; \boldE\right\} \le \frac{1}{K} \sum_{k\in \mathcal N} E_k \alpha.
 $$ 
Hence, by iterated expectation, $\text{FDR}_{\text{ep-}\cD} \leq \E(\sum_{k\in \mathcal N} E_k \alpha/K) \leq \pi_0 \alpha$.  
\end{proof}

Perhaps surprisingly, the above result does not require any assumption whatsoever about the dependence within $\boldE$. In the case of p-values that are positive regression dependent within nulls, it is natural to posit the following dependence assumption on $\boldP$ and $\boldE$ (that is intermediate in strength compared to Assumptions~\ref{assumption:coordinatewise} and~\ref{assumption:full_indep}).
\begin{assumption}
\label{assumption:pe_null_indep}
     $(P_k)_{k \in \mathcal{N}}$ is independent of $(E_k)_{k \in \mathcal{N}}$.
\end{assumption}
\begin{theorem}\label{th:ep-null-BH}
Suppose that Assumption~\ref{assumption:pe_null_indep} holds, and that $\boldP$ is positive regression dependent within nulls (Definition~\ref{definition:prdn}). Then,
\begin{enumerate*}[(i)]
    \item the ep-Simes procedure has type-I error of at most $\alpha$ under the global null hypothesis, 
    \item the ep-BH procedure has false discovery rate at most $\pi_0 \alpha \log\{e/(\pi_0 \alpha)\}$, and
    \item the ep-Hochberg procedure has family-wise error rate at most $\alpha$.
\end{enumerate*}
\end{theorem}

\begin{proof}
The result for ep-Simes follows from Theorem~\ref{th:ep-BH}: under the global null hypothesis, it holds that $\mathcal{N} = \mathcal{K}$, and so the assumptions of Theorems~\ref{th:ep-BH} and~\ref{th:ep-null-BH} are identical. Hence by definition of the Simes procedure (and noting that $\text{FWER} = \text{FDR}$ under the global null):
$$ \mathbb P\left( \text{ep-Simes rejects } \textstyle\bigcap_{k=1}^K H_k \right) = \mathbb P( \text{ep-BH rejects at least one of } H_k) = \mathrm{FDR}_{\text{ep-BH}} \leq \alpha.$$
The result on ep-BH follows from the false discovery rate linking theorem~\citep[Theorem 1]{su2018fdrlinking} which converts bounds on the false discovery rate of p-BH applied on the null p-values only to a bound on the false discovery rate of p-BH applied to all p-values. 
The false discovery rate linking theorem is also applicable to ep-BH once we interpret it as p-BH acting on $\boldPstar$. Hence it suffices to bound the false discovery rate of ep-BH applied on the null hypotheses only, and such a bound follows from Theorem~\ref{th:ep-BH}.

The result for ep-Hochberg follows from the result for ep-Simes, since the p-Hochberg procedure~\citep{hochberg1988sharper} is a shortcut for closed testing~\citep{marcus1976closed} based on p-Simes (and so ep-Hochberg is a shortcut for closed testing based on ep-Simes).
\end{proof}

\subsection{Null proportion adaptivity: the e-weighted Storey procedure (ep-Storey)}

As we explained above, weighted multiple testing procedures typically require normalized weights, that is, weights that satisfy $\sum_{k=1}^K w_k = K$. This constraint represents a fixed weight budget to be allocated across hypotheses. It is possible, however, to increase the budget in a data-driven way by adapting to the proportion of null hypotheses. For example, in the case of uniform weights (i.e., for unweighted multiple testing),~\citet*{storey2004strong} proposed to estimate the proportion of null hypotheses $\pi_0 = K_0/K$ by:
\begin{equation}
    \label{eq:storey}
    \widehat{\pi}_0 := \frac{1 + \sum_{k=1}^K \id(P_k > \tau)}{K(1-\tau)},
\end{equation}
for fixed $\tau \in (0,1)$, and then to apply the p-BH procedure with p-values $P_k$ and weights $w_k := \id(P_k \leq \tau)/\widehat{\pi}_0$. Since hypotheses with $P_k > \tau$ would be unlikely to be rejected, the procedure of \citet{storey2004strong} increases the weight budget (when $\widehat{\pi}_0 < 1$) from $K$ to $K/\widehat{\pi}_0$.

The case of null-proportion adaptive procedures was addressed by \citet{habiger2017adaptive} and \citet{RBWJ19} for arbitrary normalized weights, and by~\citet{li2019multiple} for a specific choice of data-driven weights. Here we propose  \emph{ep-Storey}, a null proportion adaptive version of ep-BH, which proceeds as follows: compute $\widehat{\pi}_0$ as in~\eqref{eq:storey} with p-values $P_k$ and then apply the weighted p-BH procedure (Definition~\ref{definition:wBH}) at level $\alpha$ with p-values $P_k$ and weights $w_k:=\id(P_k \leq \tau)E_k/\widehat{\pi}_0$. 

\begin{theorem}\label{th:ep-Storey}
Suppose that Assumption~\ref{assumption:full_indep} holds, and that $\boldP$ satisfies p-independence (Definition~\ref{definition:pindependence}). Then the ep-Storey procedure has false discovery rate at most $\alpha$.
\end{theorem}
The proof (Supplement~\ref{sec:proof-ep-storey}) proceeds as the proof of Theorem~\ref{th:ep-BH} by arguing conditionally on $\boldE$.
In defining ep-Storey, we abused terminology: ep-Storey is different than the procedure implied by Definition~\ref{definition:eweightedpvalue}, that is, applying Storey's procedure to $\boldPstar$ in which case one would estimate $\pi_0$ by \smash{$\hat{\pi}'_0 = \{1 + \sum_{k=1}^K \id(P_k > \tau E_k)\}/\{K(1-\tau)\}$} instead of~\eqref{eq:storey} and then weight hypotheses by $\id(P_k \leq \tau E_k)/\hat{\pi}'_0$.
The latter procedure controls the false discovery rate under Assumptions~\ref{assumption:coordinatewise} and~\ref{assumption:pe_is_indep} (by Theorem~\ref{th:indep}), but not necessarily under the assumptions of Theorem~\ref{th:ep-Storey}.

One shortcoming of ep-Storey occurs when some e-values are potentially very strong, say  $1/E_k \approx \alpha / K$, but the corresponding p-values satisfy $P_k > \tau$. Such hypotheses would be discarded by ep-Storey, but would be rejected by e.g., e-BH that only uses the e-values. Hence it may be advisable to use larger values of $\tau$ for ep-Storey. 

\begin{remark}
\label{rema:null_prop_adaptivity_for_weighted_procedures}
There is a further subtle benefit of ep-BH compared to weighted p-BH with normalized weights related to null proportion adaptivity. In many applications, the proportion of null hypotheses $\pi_0$ is close to $1$. In such cases, it may not be worthwhile to apply ep-Storey compared to ep-BH, since their power will be comparable ($\alpha \approx \pi_0 \alpha$) and the assumptions on the dependence of $\boldP$ are stronger in Theorem~\ref{th:ep-Storey} compared to Theorem~\ref{th:ep-BH}. For weighted p-BH with normalized weights, however, the false discovery rate is controlled at $\alpha \sum_{k \in \mathcal{N}} w_k / \sum_{k=1}^K w_k$ \citep{RBWJ19}. Hence, for an informative weight assignment that prioritizes alternative hypotheses over null hypotheses, the gains from null proportion adaptivity can be substantial even when $\pi_0 \approx 1$. In other words: the more informative the weights are, the more conservative normalized weighted p-BH becomes in terms of the false discovery rate. ep-BH does not pay such a penalty for informative e-value weights as long as $\E(E_k)=1$ for all $k \in \mathcal{N}$.
\end{remark}

\subsection{Robustness to misspecification}
\label{sec:misspecification}
So far, we have presented all results under the assumption that we have access to e-values with the property that $\E(E_k) \leq 1$ for all $k \in \mathcal{N}$ and p-values with the property $\mathbb P(P_k \leq t) \leq t$ for all $t\in [0,1]$ and all $k \in \mathcal{N}$. Nevertheless, the guarantees are often robust to some deviations from these assumptions. We consider the following possible deviations.

\emph{Inflated (anticonservative) e-values or p-values:} Suppose the e-values satisfy $\E(E_k) \leq 1+\eta$ for some $\eta >0$ and all $k \in \mathcal{N}$ or alternatively that the p-values satisfy $\mathbb P(P_k \leq t) \leq (1+\eta)t$ for all $t\in [0,1]$ and all $k \in \mathcal{N}$. Then the generalized type-I error bounds for ep-BH, ep-Bonferroni, ep-Hochberg, and ep-Simes stated above hold with $\alpha$ replaced by $(1+\eta)\alpha$.  For example, ep-BH controls the false discovery rate at level $(1+\eta)\ell_K \pi_0 \alpha$ in the setting of Theorem~\ref{th:epD_arbitrary_dependence} and at level $(1+\eta)\pi_0\alpha$ in the setting of Theorem~\ref{th:ep-BH}. (The proofs are simple and thus omitted.)

\emph{Compound e-values or p-values:} Instead of $\E(E_k) \leq 1$ holding for all $k \in \mathcal{N}$, suppose this property holds on average over all $k$~\citep{WR21,ren2022derandomized}, that is,
\begin{equation}
\label{eq:compound_evalues}
\frac{1}{K_0} \sum_{k \in \mathcal{N}} \E(E_k) \leq 1.
\end{equation}
Some results extend to this case as well. For example,
under~\eqref{eq:compound_evalues}, ep-Bonferroni controls the family-wise error rate at $\alpha$ in the setting of Theorem~\ref{th:epD_arbitrary_dependence}, and ep-BH controls the false discovery rate at $\alpha$ in the setting of Theorem~\ref{th:ep-BH}. Analogously to~\eqref{eq:compound_evalues}, one may consider compound p-values that satisfy $ \sum_{k \in \mathcal{N}} \mathbb P(P_k \leq t) \leq K_0 t$ for all $t \in [0,1]$. See~\citet{armstrong2022false} for robustness guarantees in that setting.

\section{Data-driven weighting with compound e-values}
\label{sec:unoormalized_data_driven_weighting}
\subsection{Simultaneous one-sample t-tests and means of squares as e-values}
We now turn to our second contribution:
we demonstrate the feasibility and practicality of weighted multiple testing procedures with unnormalized data-driven weights. For simplicity, we restrict attention to ep-BH and ep-Bonferroni. Our agenda will be to construct $E_k$ such that the following properties hold approximately:  $E_k$ is independent of $P_k$ for $k \in \mathcal{N}$ and~\eqref{eq:compound_evalues} holds.

Throughout this section we study the problem of conducting simultaneous
one-sample t-tests based on $n$ observations per hypothesis (and we defer extensions to simultaneous two-sample t-tests to Supplement~\ref{sec:suppl_two_sample_ttests}). To be concrete, for the $k$-th hypothesis we observe 
\begin{equation}
Y_{kj} \sim \mathrm{N}(\mu_k,\, \sigma_k^2), \text{ for } j=1,\dotsc,n,\;\; \mu_k \in \mathbb R,\;\; \sigma_k > 0,
\label{eq:gaussian_replicates}
\end{equation}
and we assume that all $Y_{kj},\, 1 \leq k \leq K,\, 1 \leq j \leq n$ are mutually independent and $n \geq 2$. We seek to test $H_k: \mu_k=0$ and to do so, we compute p-values based on the standard t-test, 
\begin{equation}
\label{eq:one_sample_ttest_pvalue}
\hat{\mu}_k := \frac{1}{n}\sum_{j=1}^n Y_{kj},\;\; \hat{\sigma}_k^2 := \frac{1}{n-1} \sum_{j=1}^n (Y_{kj} - \hat{\mu}_k)^2,\;\; T_k := \frac{\sqrt{n}\hat{\mu}_k}{\hat{\sigma}_k},
\end{equation}
and $P_k := 2\{1-F_{t, n-1}(|T_k|)\}$, 
where $F_{t, n-1}$ is the t-distribution with $n-1$ degrees of freedom.
Multiple testing with simultaneous t-tests has been studied by several authors including \citet{smyth2004linear, westfall2004weighted, finos2007fdr, bourgon2010independent,  du2014single, lu2016variance, guo2017analysis,  ignatiadis2021covariate, hoff2022smaller, ignatiadis2023empirical}. Albeit stylized, this problem has provided several new insights and 
has called attention to differences between single hypothesis testing and multiple testing.
We contribute to the literature by demonstrating how unnormalized data-driven weights alongside $\boldP$ can be constructed based on $(Y_{kj})_{k,j}$ while retaining type-I error control guarantees.

The key observation permitting the construction of data-driven weights is the following.
Let \smash{$S_k^2 := \sum_{j=1}^{n} Y_{kj}^2/n$} be the mean of squares of the observations for the $k$-th hypothesis. 
When $k \in \mathcal{N}$, that is, when $\mu_k=0$, then $S_k^2$ is complete and sufficient for $\sigma_k^2$. On the other hand, the t-statistic $T_k$ is ancillary for $\sigma_k^2$. Hence, by Basu's theorem~\citep{basu1955statistics}, it holds that $T_k$ (and so $P_k$) is independent of $S_k^2$. To summarize:
\begin{equation}
\label{eq:s_k_p_k_independence}
S_k^2\, \text{ is independent of}\; P_k \;\text{ for } k \in \mathcal{N}.
\end{equation}
Motivated by the above, \citet{westfall2004weighted} considered the weights
\begin{equation}
\label{eq:westfall_weights}
w_k := K S_k^2 \bigg / \sum_{\ell=1}^K S_{\ell}^2,
\end{equation}
and proved that weighted p-Bonferroni with p-values $P_k$ and weights $w_k$
controls the family-wise error rate. Similarly, weighted p-BH with the above p-values and weights controls the false discovery rate. 

The perspective on e-values as weights in multiple testing along with the robustness guarantee for~\eqref{eq:compound_evalues} suggests that instead of $w_k$ we could have used $E_k^* :=  K_0 S_k^2  /  \sum_{\ell \in \mathcal{N}} \sigma_{\ell}^2$, which satisfy $\sum_{k \in \mathcal{N}} \E(E_k^*) =  K_0$ (since $\E(S_k^2)=\sigma_k^2$ for $k \in \mathcal{N}$). As explained in Section~\ref{sec:misspecification}, ep-Bonferroni, resp. ep-BH with $P_k$ and $E_k^*$ control the family-wise error rate, resp. false discovery rate. 

A caveat to the above argument is that the compound e-values $E_k^*$ are not computable as they depend on the unknown data generating mechanism through \smash{$\sum_{\ell \in \mathcal{N}} \sigma_{\ell}^2$}. Our proposal then is to conservatively estimate \smash{$\sum_{\ell \in \mathcal{N}} \sigma_{\ell}^2$} by \smash{$\sum_{\ell=1}^K \hat{\sigma}_{\ell}^2$}, and to construct feasible approximations to $E_k^*$:
\begin{equation}
\label{eq:westfall_evalues}
E_k := K S_k^2 \bigg / \sum_{\ell=1}^K \hat{\sigma}_{\ell}^2.
\end{equation}
Since $\E(S_k^2) = \mu_k^2 + \sigma_k^2$, we see that by using $E_k$~\eqref{eq:westfall_evalues} in place of $w_k$~\eqref{eq:westfall_weights} we increase the expected total weight budget by approximately the factor $\sum_{k=1}^K (\mu_k^2 + \sigma_k^2) / \sum_{k=1}^K \sigma_k^2$. Our proposal furthermore controls type-I error, as the following theorem demonstrates.
\begin{theorem}
\label{th:data_driven_weights}
In the above setting, suppose there exist $\eta, \delta \in [0,1)$ such that $\mathbb P(A_{\delta}) \leq \eta$, where $A_{\delta}$ is the ``bad'' event $\{\sum_{\ell=1}^K \hat{\sigma}_{\ell}^2 < (1-\delta) \sum_{\ell \in \mathcal{N}} \sigma_{\ell}^2\}$. Then, \begin{enumerate*}[(i)]
\item the ep-Bonferroni procedure with t-test p-values $P_k$ and (approximate compound) e-values $E_k$~\eqref{eq:westfall_evalues} controls the family-wise error rate at  $\alpha':=\alpha /(1-\delta) + \eta$, and
\item the ep-BH procedure with $P_k$ and $E_k$ controls the false discovery rate at $\alpha'$.
\end{enumerate*}
\end{theorem}

\begin{proof}
We prove the result for ep-Bonferroni to highlight the core ideas
and defer the proof for ep-BH to Supplement~\ref{subsec:proof_ep_bh_data_driven}. Applying the union bound twice, we see that,
$\mathrm{FWER}_{\cD}=\p(F_\cD \ge 1) \leq \p(\{F_\cD \ge 1\}\cap A_{\delta}^c) + \eta \leq \sum_{k \in \mathcal{N}} \p(\{P_k \leq \alpha E_k /K \}\cap A_{\delta}^c) + \eta.$
Let $\tilde{E}_k := K S_k^2 / \{(1-\delta)\sum_{\ell \in \mathcal{N}} \sigma_{\ell}^2\}$, then on the complement of the event $A_{\delta}$ it holds that $E_k \leq \tilde{E}_k$. Hence we may bound $\p(\{P_k \leq \alpha E_k /K \}\cap A_{\delta}^c)$ for null $k$ by $\p(P_k \leq \alpha \tilde{E}_k /K)$ and,
$$
\begin{aligned}
\p\left( P_k \leq \frac{\alpha \tilde{E}_k}{K}\right) \stackrel{(*)}{=} \E \left [ \p \left \{ P_k \leq \frac{\alpha  S_k^2}{(1-\delta)\sum_{\ell \in \mathcal{N}} \sigma_{\ell}^2} \, \bigg | \, S_k^2  \right\} \right]
\stackrel{(**)}{\leq} \E \left \{ \frac{\alpha S_k^2}{(1-\delta)\sum_{\ell \in \mathcal{N}} \sigma_{\ell}^2}\right \}.
\end{aligned}
$$
In $(*)$ we applied iterated expectation, and in $(**)$ we used the independence in~\eqref{eq:s_k_p_k_independence}. To conclude, we recall that $\E(S_k^2) = \sigma_k^2$ for null $k$, and then we sum the above inequalities over all $k \in \mathcal{N}$.
\end{proof}
One may seek to recalibrate $\alpha$ via Theorem~\ref{th:data_driven_weights} to achieve finite-sample error control at the desired level. We do not advocate for such recalibration, since the type-I error inflation due to $\delta, \eta$ will typically be negligible, as we demonstrate next using standard concentration arguments \citep[Chapter 2.4]{boucheron2013concentration}. We postpone proof details to Supplement~\ref{sec:proof_prop_explicit_delta_eta}.
\begin{proposition}
\label{proposition:explicit_delta_eta}
In Theorem~\ref{th:data_driven_weights}, we may choose $\delta = \delta(\eta)$ for any $\eta \in (0,1)$ with $\delta(\eta) < 1$, where $\delta(\eta):= 2\{\log(1/\eta) \sum_{\ell \in \mathcal{N}} \sigma_{\ell }^4  /(n-1)\}^{1/2}  / \sum_{\ell \in \mathcal{N}} \sigma_{\ell}^2$. If  $\sigma_k \in [\ubar{\sigma}, \bar{\sigma}] $ for all $k \in \mathcal{N}$, where $0<\ubar{\sigma}\leq \bar{\sigma}$, then $\delta(\eta) \leq  2 [\log(1/\eta)\bar{\sigma}^4 / \{(n-1)K_0 \ubar{\sigma}^4\}]^{1/2}$, i.e., $\delta(\eta) = O\{\surd{\log(1/\eta)}/\surd{(nK_0)}\}$.
\end{proposition}
As an example, suppose we apply ep-BH at target level $\alpha = 0.1$ and that model~\eqref{eq:gaussian_replicates} holds with $\sigma_k^2=1$ for all $k \in \mathcal{K}$,  $\pi_0 = 0.95$, $n=10$, and $K = 20,000$. Then we can choose $\eta = 0.001$, $\delta < 0.0128$, so that the false discovery rate will be provably controlled at level at most $0.1023$.

\subsection{ On the choice of weighting function}
\label{subsec:choice_of_weighting_function}
Above we explained our results for the compound e-values $E_k \propto S_k^2$. This choice renders the results most transparent.  In the case of normalized weighting, instead of taking $w_k \propto S_k^2$, we could have proceeded with $w_k \propto \psi(S_k^2)$ for a fixed function $\psi: \mathbb R_{\geq 0} \to \mathbb R_{\geq 0}$. Due to normalization, any (fixed) choice of $\psi(\cdot)$ leads to type-I error control, but the power of the resulting procedures depends on $\psi(\cdot)$. For example,
~\citet{westfall2004weighted} also considered the choice $\psi(s) = s^{\nu}$ for fixed $\nu >0$, while~\citet{bourgon2010independent, guo2017analysis} considered $\psi(s) = \id(s > c)$ for fixed $c > 0$.\citet{ignatiadis2016data, ignatiadis2021covariate} proposed independent hypothesis weighting (IHW) which (in the present setting) uses normalized weights \smash{$w_k \propto \widehat{\psi}_{-k}(S_k^2)$}, where for any $k \in \mathcal{K}$, \smash{$\widehat{\psi}_{-k}(\cdot)$} is learned based on $(S_{\ell}^2)_{\ell \in \mathcal{K}}$ as well as a subset of the p-values that excludes $P_k$. This subset is based on a cross-fold construction (``cross-weighting'') which avoids overfitting and ensures type-I error control.

In the case of unnormalized weights, as developed in this work, it is also possible to consider $E_k \propto \psi(S_k^2)$ for more general choices of $\psi(\cdot)$. In Supplement~\ref{sec:suppl_unnormalized}, we motivate the choice
\begin{equation}
\label{eq:evalue_trunc}
E_k := K\sum_{d=0}^{6}  \frac{ n^d \Gamma(n/2)}{4^d d! \Gamma(n/2+d)} (n S_k^2)^d \Bigg /  \sum_{\ell=1}^K \sum_{d=0}^{6} \frac{n^d \Gamma\{(n-1)/2\}}{4^d d! \Gamma\{(n-1)/2+d\}} \{(n-1) \hat{\sigma}_{\ell}^2\}^d .
\end{equation}
By using~\eqref{eq:evalue_trunc} instead of~\eqref{eq:westfall_evalues}, it is possible to increase the weight budget even further, i.e., 
$\sum_{k=1}^K E_k$ can be substantially larger than $K$.
In the simulations of Section~\ref{subsec:two_sample_ttest},~\eqref{eq:evalue_trunc} is powerful, however, we leave an investigation of the optimal (potentially data-driven) choice of e-value weights to future work.

\subsection{ A hierarchy of assumptions on side-information}
\label{subsec:hierarchy}
Our discussion so far indicates a hierarchy of potential assumptions on side-information in multiple testing. Side-information in multiple testing refers to additional contextual knowledge, often in the form of covariates, that goes beyond the p-values $P_k$ and can enhance the power of multiple testing procedures.
\begin{enumerate*}[(i)]
\item A practical and widely used assumption~\citep{ignatiadis2016data, lei2018adapt, li2019multiple,  ignatiadis2021covariate} is that the side-information (in this section, $S_k^2$) is independent of the p-value $P_k$ for null $k$ (this holds by~\eqref{eq:s_k_p_k_independence} in the present setting). 
\item In other cases, the side-information can be used to compute a second p-value $P_k'$ that is independent of $P_k$. 
This is true in the meta-analysis setting, however, it is also possible in the setting of this section. \citet[Section 6.4]{du2014single} posit model~\eqref{eq:gaussian_replicates} with $\sigma_k^2 = 1$ for all $k$, and assume that the common value of all $\sigma_k^2$ is known to the analyst. They let $P_k$ be the t-test p-value as in~\eqref{eq:one_sample_ttest_pvalue} and $P_k' := 1-\chi_{n}(n S_k^2)$, where $\chi_n$ is the chi-square distribution with $n$ degrees of freedom. Then, $P_k'$ is a p-value that is independent of $P_k$ for $k \in \mathcal{N}$. 
\item Our methodological developments (e.g., Theorem~\ref{th:data_driven_weights}) indicate that there are assumptions for side-information intermediate in strength between (i) and (ii). In addition to (i), it suffices to control certain moments of the side-information averaged over all null coordinates~\eqref{eq:compound_evalues}.
\end{enumerate*}

One of the key messages of this paper is that it can be valuable to pursue (iii) through compound e-values as weights. 
In the model of this section, existing methods for multiple testing with side-information, e.g., independent hypothesis weighting, forfeit potential power gains by ignoring the distributional knowledge available for $S_k^2$ and instead using normalized weights (the upshot being that independent hypothesis weighting is applicable to more general forms of side-information). On the other hand, the assumption in \citet[Section 6.4]{du2014single} 
that $\sigma_k^2=1$ for all $k$ (and that their common value is known to the analyst) is strong: if we are willing to impose it, then it would be advisable to conduct a z-test instead of the t-test in~\eqref{eq:one_sample_ttest_pvalue}.

\section{Differential gene expression based on RNA-Seq and Microarray data}
\label{sec:data_analysis}
As a demonstration of the practicality, applicability, and power of ep-BH, we seek to detect genes that are differentially expressed in the striatum of two mice strains (C57BL/6J and DBA/2J). We use two sources of information: RNA-Seq p-values and microarray e-values.

In more detail, we use the first two experimental batches of the RNA-Seq data collected by~\citet{bottomly2011evaluating} which comprise 7 adult male mice of each of the two strains (14 mice in total). We compute p-values for differential gene expression across the two strains using \texttt{DESeq2}~\citep*{love2014moderated} adjusting for the experimental batches. Furthermore, we use microarray (Affymetrix) expression measurements collected by~\citet{bottomly2011evaluating} for 10 adult male mice (5 of each strain). p-values for such comparisons are routinely computed based on the empirical Bayes model of~\citet{lonnstedt2002replicated} with the \texttt{limma} software package~\citep{smyth2004linear}, see e.g., the workflow of~\citet{klaus2018end}. Here we demonstrate how one would proceed if e-values had been computed instead.  In Supplement~\ref{appendix:evalues_for_microarrays}, we provide a construction of  e-values under the replicated microarray model assumptions of~\citet{lonnstedt2002replicated, smyth2004linear}. The construction may be of independent interest for applications in high-throughput biology, e.g., for microarray or RNA-Seq data~\citep{ritchie2015limma}.

Our analysis leads to RNA-Seq p-values for 24,906 genes (Fig.~\ref{fig:rnaseq_microarray_histograms}A), and microarray e-values for a subset of 15,875 genes (Fig.~\ref{fig:rnaseq_microarray_histograms}B), where we map the microarray probe identifiers to Ensembl gene identifiers. We set the e-value for the remaining 9,031 genes to $1$ (which is a valid e-value). The p-values are approximately independent of all e-values since~\citet{bottomly2011evaluating} used distinct mice for the RNA-Seq and Affymetrix microarray data collection. 

\begin{figure}
    \centering
    \includegraphics[width=\linewidth]{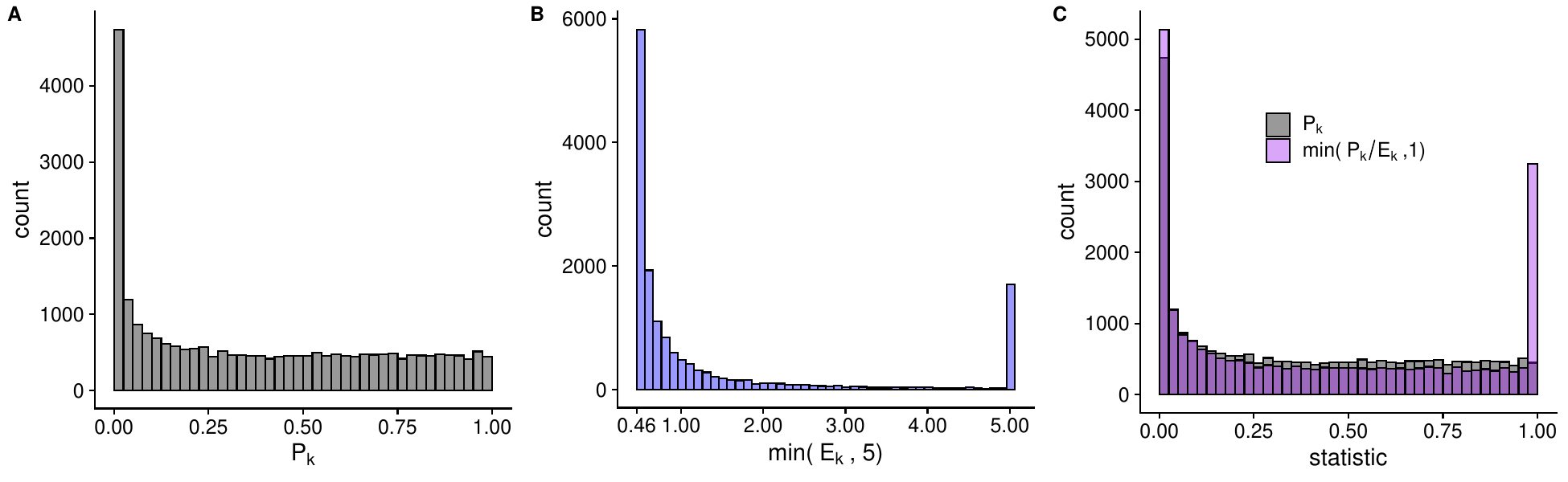}
    \caption{Differential gene expression analysis: We plot histograms of A) RNA-Seq p-values, B) microarray e-values (capped at 5), and C) $Q$-combiner p-values $Q(P_k, E_k) = (P_k/E_k) \land 1$ overlaid with the histogram of p-values from panel A).
    The smallest e-value  (panel B) is equal to $0.46$. There is an enrichment of $Q(P_k, E_k)$ compared to $P_k$ in the bin closest to $0$ (panel C).}
    \label{fig:rnaseq_microarray_histograms}
\end{figure}

We consider the following variants of weighted p-BH:
\begin{enumerate*}[(i)]
    \item Unweighted p-BH that only uses the p-values $P_k$ (and weights $w_k=1$, thus ignoring the information in the e-values).
    \item E-value weighted p-BH procedure (ep-BH) with e-values as unnormalized weights $w_k = E_k$. Fig.~\ref{fig:rnaseq_microarray_histograms}C shows a histogram of the combined p-values $(P_k/E_k) \land 1$ that are used as input for ep-BH.
    \item Weighted p-BH (wBH) with normalized e-value weights, i.e., $w_k = K\cdot E_k/\sum_{\ell=1}^K E_{\ell}$.
    \item Independent hypothesis weighting (IHW) BH~\citep{ignatiadis2021covariate, ignatiadis2016data} with p-values $P_k$ and the e-values $E_k$ as side-information (see Section~\ref{subsec:choice_of_weighting_function} for a brief description of the method). We use the implementation of independent hypothesis weighting \citep[``IHW Grenander'']{ignatiadis2021covariate} in the R/Bioconductor package ``IHW'' that stratifies hypotheses according to the side-information into $\lfloor K/1500 \rfloor$ groups of equal size.
\end{enumerate*}

We also consider two approaches that take two p-values $P_k, P_k'$ per hypothesis as input (instead of a p-value and an e-value). $P_k$ is the RNA-Seq p-value as above and $P_k'$ is the p-value returned from the microarray analysis using \texttt{limma} (Supplement~\ref{appendix:evalues_for_microarrays}).
These approaches are the 
\begin{enumerate*}[(i)]
\item[(v)] p-BH procedure applied to the Fisher combination p-values (Fisher), and
\item[(vi)] single index modulated (SIM) p-BH~\citep{du2014single}, which applies p-BH to $P_k(\hat{\theta})$, where $P_k(\theta) := \Phi\{\cos(\theta)\Phi^{-1}(P_k) + \sin(\theta)\Phi^{-1}(P_k')\}$ and $\Phi$ is the standard normal distribution function. $\hat{\theta}$ is selected as the index $\theta \in [0, \pi/2]$ that maximizes the number of discoveries of p-BH applied to $(P_k(\theta))_{k \in \mathcal{K}}$. (Single index modulated p-BH controls the false discovery rate asymptotically, but there is no finite-sample guarantee due to the data-driven choice of $\hat{\theta}$.)
\end{enumerate*}

We also consider null proportion adaptive versions of the above procedures using variants of Storey's adjustment (with $\tau=0.5$): for the unweighted method, Fisher, and single index modulation we use Storey's procedure (\citealp{storey2004strong} and~\eqref{eq:storey}), for e-value weights we use ep-Storey, for normalized e-value weights we use weighted Storey as described in~\citet{RBWJ19} and for independent hypothesis weighting we follow~\citet[Theorem 2]{ignatiadis2021covariate}. All procedures are applied to control the false discovery rate at the target level $\alpha = 0.01$.

\begin{table}
\centering
\begin{tabular}[t]{lrrr}
\toprule
  & 
 \hspace{-1cm} Avg. Weight & 90\% Weight & Discoveries\\
\midrule
\addlinespace[0.3em]
\multicolumn{4}{l}{Non-adaptive}\\
\hspace{1em}Benjamini-Hochberg (BH, unweighted) & 1.00 & 1.00 & 1973\\
\hspace{1em}E-value Weighted BH (ep-BH, our proposal) & 18.11 & 2.70 & 2387\\
\hspace{1em}Weighted BH (wBH; normalized e-value weights) & 1.00 & 0.15 & 1310\\
\hspace{1em}Indep. Hypothesis Weighted BH (IHW) & 1.00 & 2.11 & 2016\\
\hspace{1em}Fisher BH  & --- & --- & 2354\\
\hspace{1em}Single Index Modulated BH (SIM) & --- & --- & 2282\\
\addlinespace[0.3em]
\multicolumn{4}{l}{Adaptive}\\
\hspace{1em}Storey-BH (unweighted) & 1.34 & 1.34 & 2147\\
\hspace{1em}E-value Weighted Storey-BH (our proposal) & 24.35 & 3.63 & 2540\\
\hspace{1em}Weighted Storey-BH (normalized e-value weights) & 2.63 & 0.39 & 1536\\
\hspace{1em}Indep. Hypothesis Weighted Storey-BH & 1.54 & 3.49 & 2274\\
\hspace{1em}Fisher Storey-BH  & --- & --- & 2556\\
\hspace{1em}Single Index Modulated Storey-BH & --- & --- & 2479\\
\bottomrule
\end{tabular}
\caption{Multiple testing for differential gene expression based on RNA-Seq p-values and microarray e-values: The last column shows the number of discoveries of each method. The first column shows the weight budget $\sum_{k=1}^K w_k/K$. For BH-type procedures with normalized weights, this quantity is always equal to $1$. For Storey-type procedures, this budget is inflated due to accounting for the null proportion. The second column shows the upper 90\% quantile among weights $w_i$ used by each method.
\label{tab:rnaseq_plus_microarray}}
\end{table}

The results of the analysis are shown in Table~\ref{tab:rnaseq_plus_microarray}. Among the non-adaptive procedures, ep-BH makes the most discoveries, even compared to Fisher BH and single index modulated BH which have access to two independent p-values per hypothesis. The procedure that normalizes the e-value weights makes by far the least discoveries. The reason is that with the exception of the genes with the highest e-values, all other genes receive very small weights. Independent hypothesis weighting makes more discoveries than unweighted BH demonstrating that the ordering of e-values can be used to increase power, even among procedures that use normalized weights. The findings for the adaptive procedures are analogous, although in this case, Fisher Storey-BH makes the most discoveries (with a small margin compared to our ep-Storey method).

\section{Simulation study}
\label{sec:simulation}

\subsection{Evaluation}
For the simulation study we compare the same methods as in Section~\ref{sec:data_analysis}. We apply these methods at a target false discovery rate of $\alpha=0.1$. We evaluate methods in terms of their false discovery rate, and their power, which we define as  $\E\{ (R_{\mathcal{D}} - F_{\mathcal{D}})/(K-K_0)\}$.

\subsection{One sample t-test}
\label{subsec:two_sample_ttest}
We conduct a simulation study in the setting of Section~\ref{sec:unoormalized_data_driven_weighting} and generate data from model~\eqref{eq:gaussian_replicates}.
We let $n=10$, $K=20,000$, $\pi_0=0.95$ and set $\mu_k = \xi$ for the alternative hypotheses, where the effect size $\xi \in [0.5,\;1.5]$ is a varying simulation parameter. We fix $\sigma_k=1$ for all $k$. We use the e-values~\eqref{eq:evalue_trunc} for ep-BH and weighted p-BH, and the secondary p-values $P_k'$ defined in (ii) of Section~\ref{subsec:hierarchy} for Fisher and single index modulated BH. We average results over $4,000$ Monte Carlo replicates of each simulation setting.

\begin{figure}
    \centering
    \includegraphics[width=\linewidth]{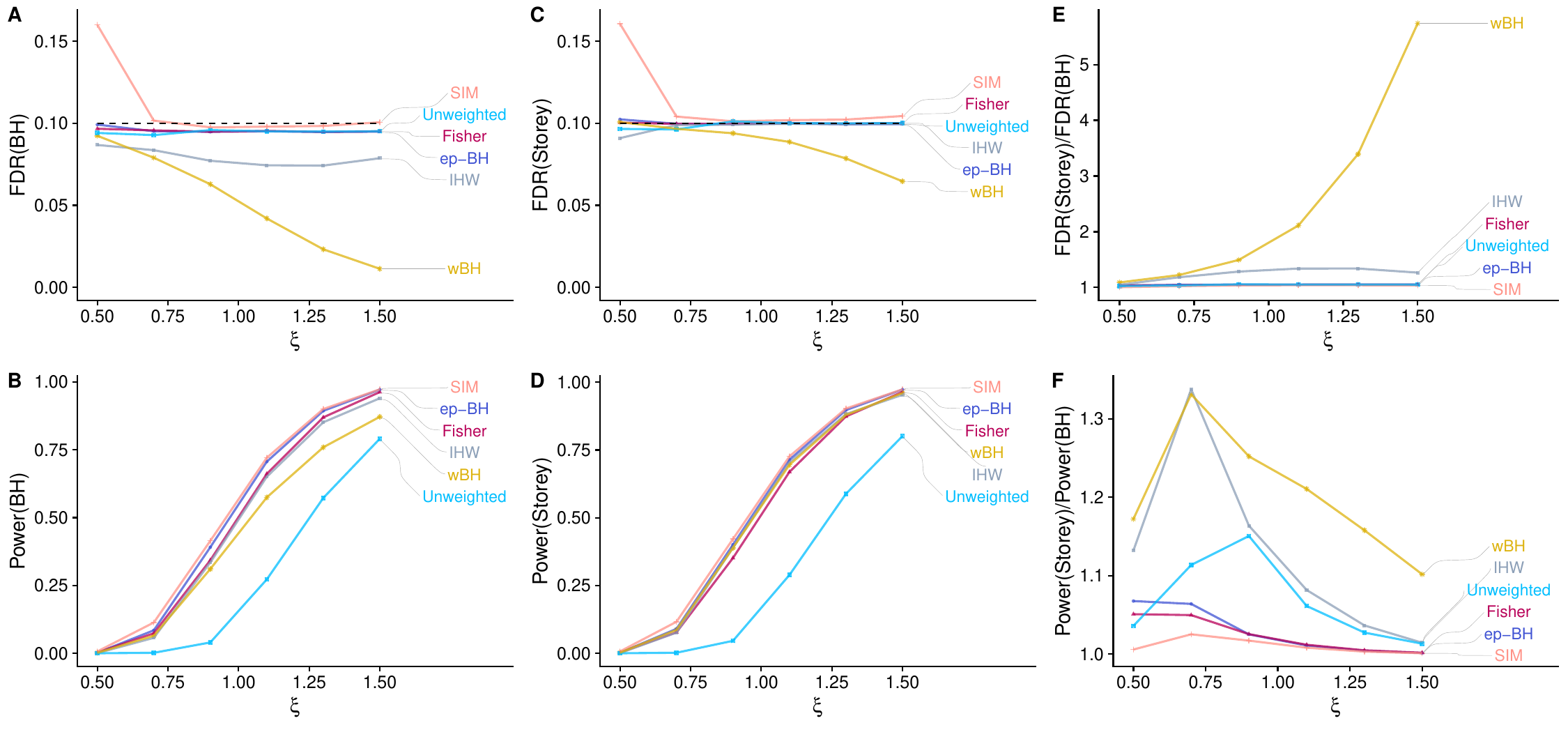}
    \caption{t-test simulations: 
    We compare non-adaptive (BH-based) methods plotting A) false discovery rate (FDR) and B) power against the effect size parameter $\xi$. All methods control the false discovery rate except single index modulated BH (SIM) at small $\xi$. Single index modulated BH is most powerful, followed by ep-BH, Fisher BH, and independent hypothesis weighting (IHW). We also evaluate the adaptive (Storey-based) counterparts of the same methods plotting their C) false discovery rate and D) power, as well the ratio of E) false discovery rate and F) power between the Storey-based and BH-based methods. Weighted p-BH (wBH) and independent hypothesis weighting benefit the most from null proportion adaptivity, cf. Remark~\ref{rema:null_prop_adaptivity_for_weighted_procedures}.
    } 
    \label{fig:ttest_sim_main}
\end{figure}

Among the weighted BH methods, all procedures control the false discovery rate (Fig.~\ref{fig:ttest_sim_main}A), and ep-BH has the most power (Fig.~\ref{fig:ttest_sim_main}B), followed by independent hypothesis weighting. Weighted p-BH (with normalized e-value weights) and unweighted p-BH have low power. Single index modulated BH has the most power, but the difference to ep-BH is small, especially considering the requirement of an additional p-value (instead of e-value) per hypothesis and that it exceeds the target false discovery rate at small $\xi$. The power of ep-BH and Fisher BH is similar. The false discovery rate, resp. power of the adaptive procedures is shown in Fig.~\ref{fig:ttest_sim_main}C, resp.~\ref{fig:ttest_sim_main}D. Figs.~\ref{fig:ttest_sim_main}E,F show the ratio of the false discovery rate and power of the adaptive methods compared to their non-adaptive counterparts. As explained in Remark~\ref{rema:null_prop_adaptivity_for_weighted_procedures}, the procedures with normalized weights derive most benefit from null proportion adaptivity. ep-BH achieves strong power gains even without null proportion adaptivity.

\begin{figure}
    \centering
    \includegraphics[width=\linewidth]{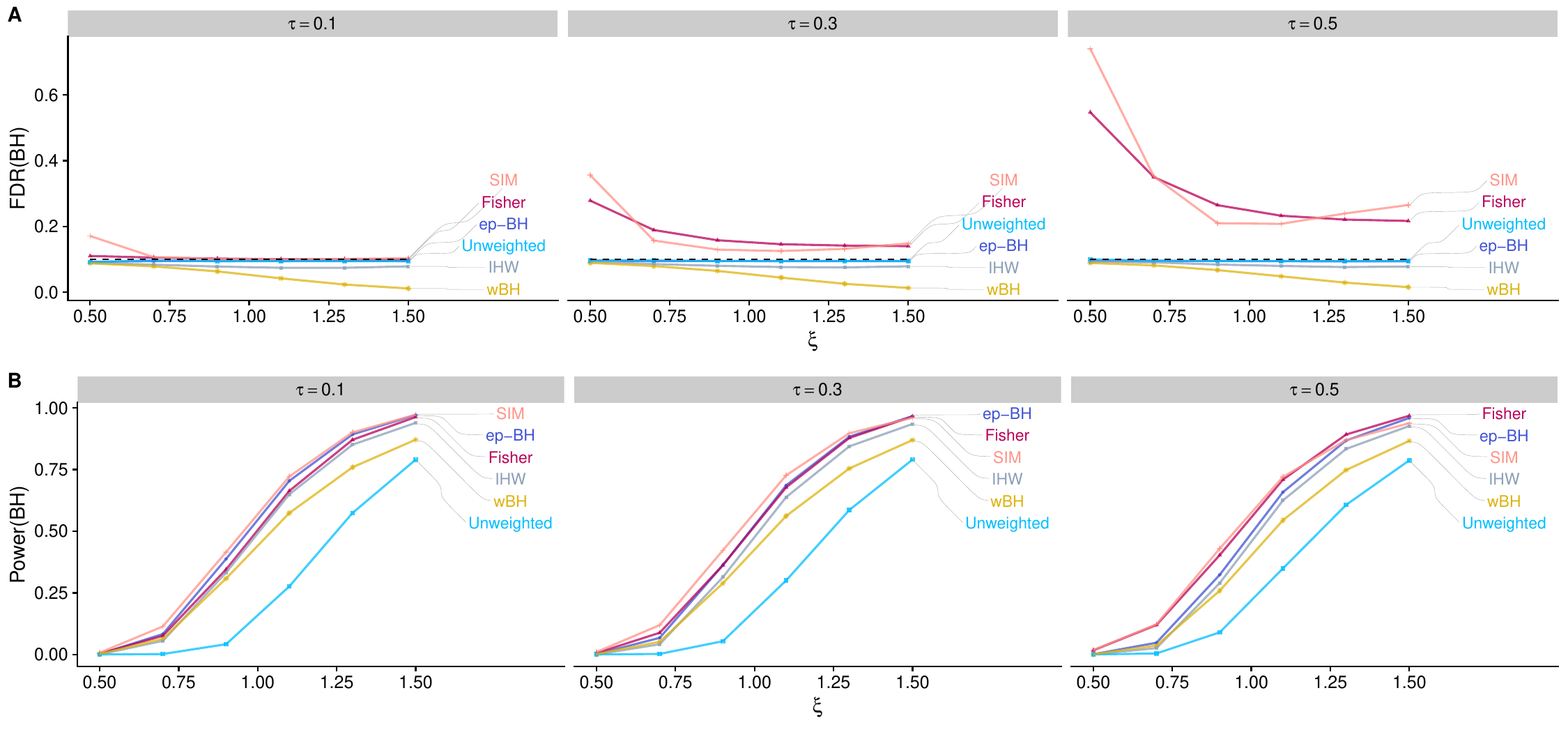}
    \caption{t-test simulations with heteroscedasticity:  Panels A) and B) are analogous to panels A and B of Fig.~\ref{fig:ttest_sim_main}. We compare non-adaptive methods in terms of their false discovery rate (FDR) and power. While in Fig.~\ref{fig:ttest_sim_main}, $\sigma_k^2=1$ for all $k \in \mathcal{K}$, here we enforce increasingly strong heteroscedasticity by drawing $\sigma_k^2 \sim \mathrm{U}(1-\tau, 1+\tau)$, where $\tau \in \{0.1, 0.3, 0.5\}$ corresponds to the different facets. 
    Fisher and single index modulated BH do not control the false discovery rate under strong heteroscedasticity $(\tau\in \{0.3,0.5\}$), while the other methods do (cf. Section~\ref{subsec:hierarchy}). ep-BH has the most power among methods controlling the false discovery rate.}
    \label{fig:ttest_sim_heterogeneous}
\end{figure}

We next tweak the simulation and introduce heteroscedasticity by drawing $\sigma_k^2 \sim \mathrm{U}(1-\tau, 1+\tau)$, where $\tau \in \{0.1,0.3,0.5\}$ is a simulation parameter. In Fig.~\ref{fig:ttest_sim_heterogeneous} we show results for the non-adaptive methods only. The weighted BH methods perform similarly as in the homoscedastic case of Fig.~\ref{fig:ttest_sim_main}. Methods that take two p-values as input (Fisher and single index modulated BH) strongly violate the target false discovery rate for $\tau \in \{0.3, 0.5\}$. The reason is that the secondary p-value is no longer a valid p-value when the assumption $\sigma_k^2=1$ is violated (cf. Section~\ref{subsec:hierarchy}).

\subsection{Combining RNA-Seq and microarray data}
\label{sec:sim_rnaseq_microarray}
We next consider a simulation that mimics the RNA-Seq/microarray application of Section~\ref{sec:data_analysis}. We simulate datasets with $K=10,000$ genes of two-sample comparisons with 20 samples for each combination of condition (control/treatment) and technology (RNA-Seq/microarray). For the synthetic RNA-Seq datasets, inspired by the simulation setup in~\citet{love2014moderated}, we generate negative binomial count data with mean and dispersion parameters chosen to approximate realistic moments by resampling from the joint distribution of mean/dispersion parameters of the simulations in~\citet{love2014moderated}, truncated to mean values of at least 1. We let $\pi_0 = 0.8$, and sample the alternative genes uniformly among all genes and then set the (binary) logarithmic fold changes of the treated samples to $+\xi$, resp. $-\xi$ with probability $1/2$, where $\xi \in [0.3, 0.9]$ is a varying simulation parameter. To generate synthetic microarray data, we follow the model described in Supplement~\ref{appendix:evalues_for_microarrays}. We first sample variances $\sigma_k^2$, $k \in \mathcal{K}$ from~\eqref{eq:limma_shrinkage_assumption} with $\nu_0=3.64$ and $s_0^2=0.0144$ (which are the estimates of $\nu_0$ and $s_0^2$ in the microarray data described in Section~\ref{sec:data_analysis}). We then order the variances according to the order of the mean counts from the RNA-Seq simulation (i.e., the gene with largest mean count in the RNA-Seq experiment also has the largest variance in the microarray experiment). For all null genes, the effect size is $\beta_k =0$, while for differentially expressed genes we let $\beta_k \sim (1-\pi_M) \delta_0 + \pi_M \mathrm{N}(0,\, 0.5 \sigma_k^2)$, where $\delta_0$ is a point mass at $0$ and $\pi_M \in \{1, 0.5, 0\}$ is a simulation parameter. In words, when $\pi_M = 0$, the microarray dataset is completely uninformative, while when $\pi_M=1$ all differentially expressed genes in the RNA-Seq dataset are also differentially expressed in the microarray dataset. We generate summary statistics for the 20 vs. 20 comparisons for each gene as in~\eqref{eq:limma_distribution_assumptions}. Finally, we compute p-values and e-values as in Section~\ref{sec:data_analysis}. Results are averaged over 100 Monte Carlo replicates.

\begin{figure}
    \centering
    \includegraphics[width=\linewidth]{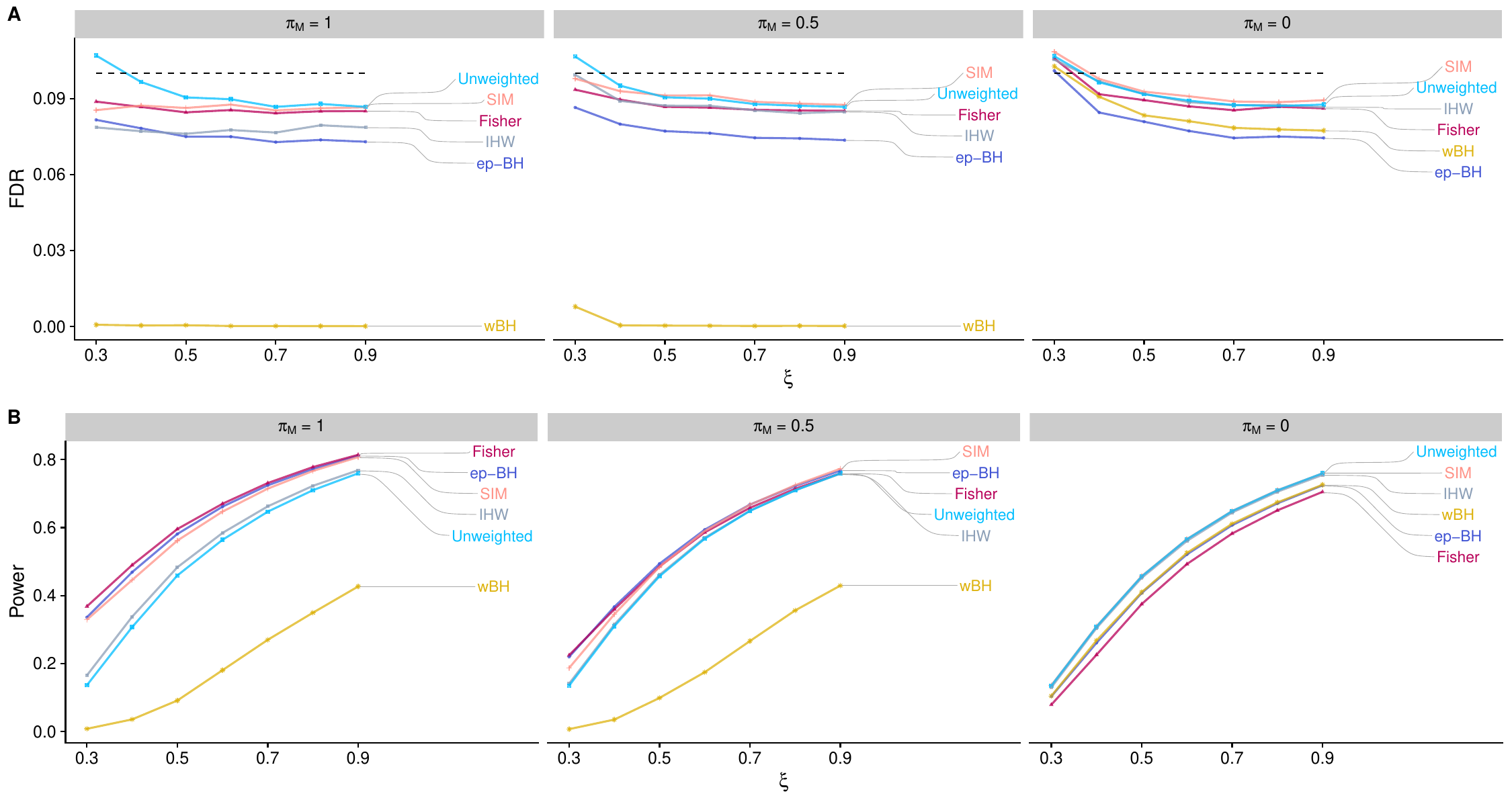}
      \caption{RNA-Seq and microarray meta-analysis simulation: We compare non-adaptive (BH-based) methods plotting A) false discovery rate (FDR) and B) power against the effect size $\xi$ and against the informativeness of the microarray data (parameter $\pi_M$ in the facets). When $\pi_M=1$, the microarray data are fully informative and Fisher BH has the most power followed by ep-BH. weighted p-BH (wBH) has less power than even unweighted p-BH. When the microarray data are completely uninformative $(\pi_M=0)$, then unweighted p-BH, independent hypothesis weighting (IHW), and single index modulated BH (SIM) have the most power, and Fisher BH has the least power. At intermediate informativeness ($\pi_M = 0.5$), weighted p-BH has the least power.
      }
    \label{fig:rnaseq_sim_main}
\end{figure} 
All non-adaptive methods control the false discovery rate (Fig.~\ref{fig:rnaseq_sim_main}A). When $\pi_M =1$, Fisher BH has the most power  (Fig.~\ref{fig:rnaseq_sim_main}B), followed by ep-BH, single index modulated BH, and independent hypothesis weighting. ep-BH more than doubles the power compared to unweighted p-BH at the smallest values of the logarithmic fold change $\xi$. Weighted p-BH has less power than unweighted p-BH and its false discovery rate is almost $0$. The case $\pi_M =0$ is chosen as a challenging setting for ep-BH with completely uninformative e-values. In this case, ep-BH and weighted p-BH have similar power, which is only slightly lower than the power of unweighted p-BH, single index modulated BH, and independent hypothesis weighting. Single index modulated BH, and independent hypothesis weighting essentially collapse to unweighted p-BH when the second p-value (resp. e-value) is uninformative. Fisher BH has the lowest power in this case. For the intermediate choice of $\pi_M = 0.5$, single index modulated BH, ep-BH, and Fisher BH have the most power. Supplementary Fig.~\ref{fig:rnaseq_sim_suppl} shows the false discovery rate and power of the adaptive procedures and
the overall take-home message remains the same: when e-values are informative, ep-BH can lead to substantial and practical power gains while maintaining type-I error guarantees.

\section*{Supplementary material}
\label{SM}
Supplementary material includes omitted proofs, methodological details (e.g., the first (minimally) adaptive e-BH procedure, inspired by \citet{SG17} in Supplement~\ref{sec:minimally_adaptive}), and additional simulation results.
All numerical results of this paper are fully third-party reproducible, and we provide the code on Github:  \texttt{https://github.com/nignatiadis/evalues-as-weights-paper}.

\bibliographystyle{plainnat}
\bibliography{local.bib}

\clearpage

\renewcommand{\thefigure}{S\arabic{figure}}
\setcounter{figure}{0}

\renewcommand{\thetable}{S\arabic{table}}
\setcounter{table}{0}
\renewcommand{\theequation}{S\arabic{equation}}
\setcounter{equation}{0}

\renewcommand{\thesection}{S\arabic{section}}
\setcounter{section}{0}

\renewcommand{\thepage}{S\arabic{page}}
\setcounter{page}{1}

\section{Omitted proofs}
\label{sec:omitted_proofs}
\subsection{Proof of Theorem~\ref{th:combiners}}
\label{sec:proof-admissible}

 Let $P$ be a p-value and $E$ be an e-value. They are assumed independent in (i) and (ii) below. 
For a fixed $e\in [0,\infty)$, we will frequently rely on a specific  distribution $F_e$ of e-values given by, for $X\sim F_e$,
 $\p( X  =e)=1/e =1- \p(X =0)$ if $e\ge 1$
and $\p(X   =e)=1/2 =  \p(X =2-e)$ if $e< 1$.
It is clear that $\E(X)=1$. 
 \begin{enumerate}[(i)]
 \item We have $\E\{h(P)E\}\le 1$ since $h(P)$ is an e-value independent of $E$. Hence, $\Pi_h$ is an i-pe/e combiner. 
To show its admissibility,  
suppose for the purpose of contradiction that an i-pe/e combiner $f$ satisfies 
$f\ge \Pi_h$ and  
$f (p,e)> \Pi_h(p,e)$ for some $(p,e)\in [0,1]\times \overline{\R}_+$. 
Clearly, $e\in [0,\infty)$ and $p \in (0,1]$. 
Since  $h$ is upper semicontinuous and $q\mapsto f(q,e)$ is decreasing, 
there exists $p'<p$ such that 
$f (q,e) \ge f(p,e)> \Pi_h(p',e)\ge \Pi_h(q,e)$ for all $q\in [p',p]$.
Take $P\sim \mathrm {U}(0,1)$ and $E\sim F_e$.  
Since $\E\{\Pi_h(P,E)\}=1$,  $f\ge \Pi_h$,
and $f (q,e)> \Pi_h(q,e)$ for all $q\in [p',p]$,
we have $\E\{f(P,E)\}>\E\{\Pi_h(P,E)\}=1$ which means that $f(P,E)$ is not an e-value, contradicting the fact that $f$ is an i-pe/e combiner. 
This contradiction shows that $\Pi_h$ is admissible. 
 \item For $\alpha \in (0,1)$, we have  $\p\{Q(P,E)\le \alpha\}= \p(P\le \alpha E)   = \E\{\p(P\le \alpha E|E)\}  \le \E(\alpha E)  \le \alpha.$
 Therefore, $Q$ is an i-pe/p combiner.  To show its admissibility,  
suppose for the purpose of contradiction that an i-pe/p combiner $f$ satisfies 
$f\le Q$ and  
$f (p,e)< (p/e)\wedge 1$ for some $(p,e)\in [0,1]\times \overline{\R}_+$. 
Since $a\mapsto f(p,a)$ is decreasing,  
we can assume  $e\in [p,\infty)$ by replacing $e$ with $p$ if $e<p$. Take $P\sim \mathrm {U}(0,1)$.  
Since $q\mapsto f(q,e)$ is increasing, 
there exists $p'<p$ such that 
$f (q,e) \le f(p,e) < p'/e$ for all $q\in [0,p]$. 
This gives
$\p\{f(P,e)  \le p'/e\} \ge  \p(P\le p)=p.$ 
For $\alpha = p'/e \in (0,1)$,  if $e\ge 1$, then take $E\sim F_e$, so that: 
\begin{align*}
\p\{f(P,E)  \le \alpha \}& = 
\p\{f(P,e)\le p'/e\} e^{-1}   + \p\{f(P,0) \le \alpha\}(1-e^{-1})
  \ge p /e >\alpha.
\end{align*}
If $e<1$, then take $E$ distributed such that $\p(E   =e)=\lambda$ and $  \p\{E = (1-\lambda e)/(1-\lambda)\} = 1-\lambda$ with $\lambda \in (0,1)$ chosen sufficiently small, so that $\alpha (1-\lambda e)/(1-\lambda) < 1$. Then:
\begin{align*}
\p\{f(P,E)  \le \alpha \}& = \lambda 
\p\{f(P,e)\le p'/e \} + (1-\lambda)   \p[f\{P,(1-\lambda e)/(1-\lambda)\} \le \alpha] 
\\ & \ge  \lambda p +   (1-\lambda)   \p\{ P  \le \alpha (1-\lambda e)/(1-\lambda)\} 
\\& =  \lambda p + (1-\lambda)\alpha (1-\lambda e)/(1-\lambda) \\ 
& = \lambda p + (1-\lambda e)p'/e \\ 
&= \lambda (p -p') + p'/e > \alpha.
%\frac{1}{2} \left\{p + (2-e) p'/e\right\}  =\frac{1} {2} (p-p') +   p'/e 
%>\alpha.
\end{align*}
Hence, $f(P,E)$  is not a p-value, and this contradicts the fact that $f$ is an i-pe/p combiner. 
This contradiction shows that $Q$ is admissible. 

 \item The weighted average of two arbitrary e-values is an e-value; hence $M^\lambda_{h}$ is a pe/e combiner. Its admissibility follows essentially the same proof as part (i), which we do not repeat. 
 \item Since $1/E$ is a p-value, 
 the Bonferroni combination of $P$ and $1/E$, $2 \min (P,1/E)$, is a p-value, and hence $B$ is a pe/p combiner. 
  To show its admissibility,  
suppose for the purpose of contradiction that a pe/p combiner $f$ satisfies 
$f\le  B$ and  
$f (p,e)< \{2(p\wedge e^{-1})\wedge 1\}$ for some $(p,e)\in [0,1]\times \overline{\R}_+$. 
By monotonicity of $f$, we can  increase $e$ to $1/p$ or decrease $p$ to $e^{-1} \wedge (1/2)$,
and this does not change the value of $\{2(p\wedge e^{-1})\wedge 1\}.$ Hence, we can assume that $f(p,1/p)< 2p$ for some $p\in (0,1/2]$ by noting that $f(p,1/p)< 2p$ automatically holds for $p>1/2$ because $B\le 1$. 
Since $q\mapsto f(q,e)$ is increasing, 
there exists $p'<p$ such that 
$f (q,1/p) \le f(p,1/p) < 2p'$ for all $q\in [0,p]$.  
Take $P\sim \mathrm {U}(0,1)$ and define
$ E = (1/p)\id(P\in [p',p)) +  (1/p')\id(P\in [p,p + p'p'/p]).$ Then $\E(E) = (p-p')/p + p'/p = 1$. Next, consider the following three cases:
\begin{enumerate*}[1)]
\item when $P \leq p'$, then $f(P, E) \leq 2p'$,
\item when $P\in [p,p + p'p'/p]$, then $E \leq 1/p'$ and so $f(P, E) \leq 2p'$, 
\item when $P \in (p',p)$, then $f(P, E) \leq f(p, 1/p) \leq 2p'$.
\end{enumerate*}
Hence $f(P,E) \leq 2p'$ on the event $\{P \leq p + p'p'/p\}$, and so:
$$\p\{f(P,E) \leq 2p'\} \geq \p(P \leq p + p'p'/p) = p + p'p'/p = p'( p/p' + p'/p ) > 2p'.$$
Hence, $f(P,E)$  is not a p-value, and this
contradicts the fact that $f$ is a pe/p combiner. 
This contradiction shows that $B$ is admissible.  
\end{enumerate}

\subsection{Proof of Theorem~\ref{th:ep-Storey}}
\label{sec:proof-ep-storey}

As in the proof of Theorem~\ref{th:ep-BH} we start by arguing conditionally on $\boldE$. The proof of~\citet[Theorem 1a]{RBWJ19} yields that:
$$\E\left\{\frac{F_{\cD}}{R_{\cD} }\mid \boldE \right\} \le \frac{\alpha }{K} \sum_{k\in \mathcal N} E_k \E\left\{\frac{1}{\widehat{\pi}_{0}^{-k}} \mid \boldE \right\},\;\; \text{where }\; \widehat{\pi}_{0}^{-k} = \frac{1 + \sum_{j \neq k} \id_{\{P_j > \tau\}}}{K(1-\tau)}.$$
By~\citet[Lemma 3]{RBWJ19}, it follows that $\E\{1/\widehat{\pi}_{0}^{-k} \mid  \boldE\} \leq 1/\pi_0$, and hence,
  $$
\E\left(\frac{F_{\cD}}{R_{\cD} } \right) = \E\left [  \E\left \{ \frac{F_{\cD}}{R_{\cD} }\mid  \boldE  \right\} \right ] \le\frac \alpha  K \E\left (  \sum_{k\in \mathcal N} E_k \frac{K}{K_0}\right )  \le  \alpha,
 $$  
 as claimed.

\subsection{Proof for ep-BH with data-driven weights (Theorem~\ref{th:data_driven_weights})}
\label{subsec:proof_ep_bh_data_driven}

\begin{proof}
We now prove the FDR control guarantee for ep-BH with data-driven e-values. We call the procedure $\mathcal{D}$.  Let
$\hat{c} := \sum_{\ell =1}^K \hat{\sigma}_{\ell}^2,$ so that $E_k = K S_k^2 / \hat{c}$.

It will be helpful to note the following standard decomposition: 
$$\sum_{j=1}^n Y_{kj}^2 = \sum_{j=1}^n (Y_{kj} - \hat{\mu}_k)^2 + n \hat{\mu}_k^2\;\; \Longrightarrow\;\; n S_k^2 = (n-1)\hat{\sigma}_k^2 + n \hat{\mu}_k^2.$$
Dividing by $\hat{\sigma}_k^2$ yields that $n S_k^2 / \hat{\sigma}_k^2 = (n-1) + T_k^2$, i.e., 
$$\hat{\sigma}_k^2 = \frac{n S_k^2}{(n-1) + T_k^2}.$$
Let us write $\hat{c}$ as a function of $S_k^2, T_k^2$, $k \in \mathcal{K}$:
$$\hat{c} = \frac{1}{K}\sum_{\ell=1}^K \frac{n S_{\ell}^2}{(n-1) + T_{\ell}^2}.$$
Finally, noting that $P_k$ is a strictly decreasing function of $T_k^2$, we see that
$$ \hat{c} = h(\boldP, \boldSsq),$$
where $h(\cdot)$ is a function $[0,1]^K \times \R_{> 0}^K \to \R_{>0}$. For fixed $\boldSsq$, $h(\cdot, \boldSsq)$ is increasing in $\boldP$.

By the preceding arguments, we see that in fact we may interpret the whole multiple testing procedure $\cD$ as a function of $\boldP, \boldSsq$, i.e., $\cD = \cD(\boldP, \boldSsq)$. Furthermore, for fixed $\boldSsq$, the number of rejections of $\cD$ are decreasing in $\boldP$---this follows by standard arguments for (weighted) p-BH along with the monotonicity established for $\hat{c}$.

Let us also define $c^* := \sum_{\ell \in \mathcal{N}} \sigma_{\ell}^2$ and also
$\tilde{c} := \max\{ \hat{c}, (1-\delta)c^*\}.$
Furthermore let $\cD'$ be the ep-BH procedure with (approximate) e-values $\tilde{E}_k = K S_k^2 / \tilde{c}$ (instead of $K S_k^2 / \hat{c}$).\footnote{An analogous proof strategy is pursued in~\citet[Lemma 4.3]{BR08}.}  Notice that $\cD$ and $\cD'$ are identical on the event $A_{\delta}^c$, i.e., on the complement of the event $A_{\delta}$. Furthermore, $\cD'$ and $\tilde{c}$ inherit the monotonicity properties that we established for $\cD$ and $\hat{c}$ above. We argue that it suffices to study $\cD'$:
$$ \text{FDR}_{\cD} = \E\left( \frac{F_{\cD}}{R_{\cD}}\right) =\E\left\{ \frac{F_{\cD}}{R_{\cD}} \id(A_{\delta}) \right\} + \E\left\{\frac{F_{\cD}}{R_{\cD}} \id(A_{\delta}^c) \right\} \leq \mathbb P(A_{\delta}) + \E\left( \frac{F_{\cD'}}{R_{\cD'}}\right).$$
By assumption, it holds that $\mathbb P(A_{\delta}) \leq \eta$, and so it suffices to bound the second term.

Let us call $\mathcal{O} = \{ (S_k^2)_{k \in \mathcal{K}}, (P_k)_{k \in \mathcal{K}\setminus \mathcal{N}}\}$. By~\eqref{eq:s_k_p_k_independence} and our assumption of joint independence of all the $Y_{kj}$, it follows that $(P_k)_{k \in \mathcal{N}}$ is distributed as $U[0,1]^{K_0}$ conditionally on $\mathcal{O}$. For $k \in \mathcal{N}$:
$$
\begin{aligned}
\E\left\{\frac{\id(H_k \text{ rejected})} {R_{\cD'}}\right\} &\leq  \E\left\{\frac{\id\left(P_k \leq \frac{\alpha R_{\cD'} S_k^2}{ \tilde{c}}  \right)} {R_{\cD'}}\right\}
 = \E\left[ \E\left\{\frac{\id\left(P_k \leq \frac{\alpha R_{\cD'} S_k^2}{ \tilde{c}}  \right)} {R_{\cD'}} \Bigg | \mathcal{O} \right\}\right]\\ 
&\stackrel{(*)}{\leq} \E\left( \frac{ \alpha S_k^2} { \tilde{c}} \right)  
 \leq \frac{\alpha}{ (1-\delta)c^*}\E(S_k^2) 
 = \frac{\alpha}{(1-\delta)} \cdot \frac{\sigma_k^2}{ \sum_{\ell \in \mathcal{N}} \sigma_{\ell}^2}.
\end{aligned}
$$
The crucial argument is $(*)$. Herein 
we applied the superuniformy lemma of~\citet[Lemma 1(b)]{RBWJ19} conditionally on $\mathcal{O}$.
Summing over all $k \in \mathcal{N}$, we conclude.
\end{proof}

\subsection{Proof of Proposition~\ref{proposition:explicit_delta_eta}}
\label{sec:proof_prop_explicit_delta_eta}
\begin{proof}
The crux of the argument is that the left tail of a gamma random variable is sub-Gaussian. In particular, let $X \sim \Gamma(a,b)$ for $a, b>0$ (where $a$ is the shape and $b$ is the scale). Then, e.g., by~\citet[Chapter 2.4]{boucheron2013concentration}:
$$\E(\exp[\lambda \{X - \E(X)\}]) \leq \exp( \lambda^2 a b^2 / 2) \;\text{ for any }\; \lambda < 0.$$
Next notice that for any $k \in \mathcal{K}$, it holds that $\hat{\sigma}_k^2 \sim \Gamma\{(n-1)/2,\, 2\sigma_k^2/(n-1)\}$ and $\E(\hat{\sigma}_k^2) = \sigma_k^2$. Hence, using independence across $k$, it follows that:
$$\E\left(\exp\left[\lambda \left\{\sum_{\ell=1}^K (\hat{\sigma}_{\ell}^2 - \sigma_{\ell}^2) \right\}\right]\right) \leq \exp\left( \frac{\lambda^2}{2} \frac{2}{n-1}\sum_{\ell=1}^K \sigma_{\ell}^4 \right) \;\text{ for any }\; \lambda < 0.$$
By a standard Chernoff argument (applied to the left tail), this implies that for any $\varepsilon > 0$:
\begin{equation}
\label{eq:left_tail_hoeffding}
\p\left\{ \sum_{\ell=1}^K (\hat{\sigma}_{\ell}^2 - \sigma_{\ell}^2)  < -\varepsilon\right\} \leq \exp\left\{ -  \frac{\varepsilon^2 (n-1)}{4 \sum_{\ell=1}^K \sigma_{\ell}^4}\right\}.
\end{equation}
We may upper bound the probability of the ``bad event'' $A_{\delta}$ as follows:
$$ 
\begin{aligned}
\p(A_{\delta}) = \p\left\{\sum_{\ell=1}^K \hat{\sigma}_{\ell}^2 < (1-\delta) \sum_{\ell \in \mathcal{N}} \sigma_{\ell}^2\right\} \leq  \p\left\{ \sum_{\ell \in \mathcal{N}} (\hat{\sigma}_{\ell}^2 - \sigma_{\ell}^2) < -\delta \sum_{\ell \in \mathcal{N}} \sigma_{\ell}^2  \right\}.
\end{aligned}
$$
Now let $\varepsilon = \delta \sum_{\ell \in \mathcal{N}} \sigma_{\ell}^2$. Then, so that $\p(A_{\delta}) < \eta$ for $\eta \in (0,\,1)$, it suffices to choose $\varepsilon$ such that the right hand side in~\eqref{eq:left_tail_hoeffding} (applied only to the nulls) is less than $\eta$, i.e., we may choose any
$$\varepsilon \geq 2  \left\{\log(1/\eta) \frac{\sum_{\ell\in \mathcal{N}} \sigma_{\ell}^4}{n-1}\right\}^{1/2}.$$
Rearranging in terms of $\delta$, it thus suffices that:
\begin{equation}
\label{eq:lower_bound_on_delta}
\delta \geq   2\left\{\log(1/\eta) \frac{\sum_{\ell\in \mathcal{N}} \sigma_{\ell}^4}{n-1}\right\}^{1/2} \Big / \sum_{\ell \in \mathcal{N}} \sigma_{\ell}^2 =: \delta(\eta).
\end{equation}
The specific form of $\delta(\eta)$ announced in the statement of the proposition follows from the requirement that $\delta \in [0, 1)$.

Now suppose that $\sigma_k \in [\ubar{\sigma}, \bar{\sigma}] $ for all $k \in \mathcal{N}$. Then:
$$\delta(\eta) \leq   2  \left\{\log(1/\eta) \frac{K_0 \bar{\sigma}^4}{n-1}\right\}^{1/2} \Big / (K_0 \ubar{\sigma}^2) =  2  \left\{\log(1/\eta) \frac{ \bar{\sigma}^4}{(n-1)K_0 \ubar{\sigma}^4}\right\}^{1/2}.$$
Let us also highlight at this point that the above bound does not depend at all on the configuration of the variances of the alternative hypotheses. 
\end{proof}

\section{Results for additional multiple testing procedures}
\label{sec:additional_procedures}

\subsection{Definition of additional procedures}
\begin{definition}[p-Holm procedure (\citealp{holm1979simple})]
\label{definition:hHolm}
For $k\in \mathcal K$, let $P_{(k)}$ be the $k$-th order statistic of the p-values $P_1,\ldots,P_K$, from the smallest to the largest.
The p-Holm procedure rejects the $k_{\text{hm}}^*$ hypotheses with the smallest p-values, where
$$k_{\text{hm}}^* := \max\left\{k\in \mathcal K\,:\, P_{(j)} \le \frac{\alpha}{K-j+1}\, \text{ for all }\, j =1,\dotsc,k\right\},$$
with the convention $\max(\varnothing)=0$.
\end{definition}
The p-Holm procedure controls the family-wise error rate under arbitrary dependence between the p-values. The p-Holm procedure may be derived as the closed testing procedure based on the p-Bonferroni test.

\begin{definition}[p-Hommel procedure (\citealp{hommel1988stagewise})]
\label{definition:pHommel}
For $k\in \mathcal K$, let $P_{(k)}$ be the $k$-th order statistic of the p-values $P_1,\ldots,P_K$, from the smallest to the largest.
The p-Hommel procedure computes
$$k_{\text{hl}}^* := \max \left\{ k \in \mathcal{K}\,:\, P_{(K-k + i)} > \frac{i\alpha}{k}\, \text{ for all }\, i =1,\dotsc,k \right\},$$
with the convention $\max(\varnothing)=0$. If $k_{\text{hl}}^*=0$, then the p-Hommel procedure rejects all hypotheses, otherwise it rejects all hypotheses with $P_k \leq \alpha / k_{\text{hl}}^*$.
\end{definition}
The p-Hommel procedure controls the family-wise error rate when $\boldP$ is positive regression dependent within nulls. Furthermore, the p-Hommel procedure is the exact closed testing procedure based on the p-Simes test (while p-Hochberg is a shortcut).

\subsection{Results for ep-\texorpdfstring{$\mathcal{D}$}{D} procedures}

By plugging in the p-Hommel procedure (Definition~\ref{definition:pHommel}) into Definition~\ref{definition:eweightedpvalue} we get the ep-Hommel procedure. Analogously, by plugging in the p-Holm procedure (Definition~\ref{definition:hHolm}) into Definition~\ref{definition:eweightedpvalue}, we get the ep-Holm procedure.

Suppose the assumptions of Theorem~\ref{th:epD_arbitrary_dependence} hold. Then ep-Holm controls the family-wise error rate at level $\alpha$ (since p-Holm controls the family-wise error rate under arbitrary p-value dependence).

Suppose the assumptions of Theorem~\ref{th:ep-null-BH} hold. Then, the ep-Hommel procedure controls the family-wise error rate at level $\alpha$; the proof is entirely analogous to the proof for ep-Hochberg.

\section{Multiple testing with the i-pe/e combiner \texorpdfstring{$\Pi_h$}{}}
\label{sec:mtp_with_pi_h}
In this supplement we provide further results on multiple testing with e-values and p-values that go beyond the $Q$-combiner. Throughout we assume that Assumption~\ref{assumption:coordinatewise} holds, that is, $P_k$ is independent of $E_k$ for $k \in \mathcal{N}$. Hence, under this assumption we can combine $P_k$ and $E_k$ with the admissible i-pe/e combiner $\Pi_h$ to get e-values $E^*_k = h(P_k)E_k$.

\subsection{The pe-BH procedure for false discovery rate control}
\label{sec:peBH}
The i-pe/p combiner $\Pi_h$ and the e-BH procedure motivate the following procedure as an alternative to ep-BH:
\begin{definition}[p-weighted e-BH procedure (pe-BH)]
Choose $h\in \mathcal C^{\rm p/e}$. For $k\in \mathcal K$, compute $E^*_k= h(P_k) E_k$ by applying the i-pe/e merger $\Pi_h$,
and then supply $(E^*_1,\ldots,E^*_K)$ to the e-BH procedure at level $\alpha$.
\end{definition}
We immediately have the following result:
\begin{theorem}\label{th:pe-BH}
Suppose that Assumption~\ref{assumption:coordinatewise} holds. Then,
the pe-BH procedure has false discovery rate at most $\alpha K_0/K$ .
\end{theorem}

\begin{proof}
The theorem follows 
by combining the fact that $E_1^*,\ldots,E^*_K$ are e-values for $H_1,\ldots,H_K$ due to Theorem \ref{th:combiners} 
and the false discovery rate guarantee  $\alpha K_0/K$ of the e-BH procedure in  \citet[Theorem 2]{WR21}. 
We emphasize that no dependence assumption on either $\boldP$ or $\boldE$ is required. 
\end{proof}

We now contrast the pe-BH procedure to the ep-BH procedure.
In the ep-BH procedure, e-values are used as weights for the p-values. Intuitively, if $E_k>1$, then there is some evidence against  $H_k$ being a null, and we have $P_k/E_k<P_k$ (assuming $P_k\ne 0$); that is, the weight strengthens the signal of $P_k$. Conversely, if $E_k<1$, then there is no evidence against $H_k$ being a null, and  we have  $P_k/E_k>P_k$. The above interpretation of e-values as weights is quite natural. The situation for the pe-BH procedure, where p-values are used as weights for the e-values, is somewhat different. 
For simplicity, suppose that we use the calibrator $h\in \mathcal C^{\rm p/e}$ given by $h(p)=p^{-1/2}-1$. It is clear that $h(p)>1$ if and only if $p<1/4$. Hence, the signal of the e-value $E_k$ will be strengthened in case $P_k<1/4$. This is not surprising as observing a p-value in $(0.25,1)$ generally does not indicate evidence against the null. Other choices of $h\in \mathcal C^{\rm p/e}$ lead to different thresholds, and this is consistent with the fact that there is no universal agreement on which moderate values of a p-value should be considered as carrying some (weak) evidence against the null.

In terms of power, the ep-BH procedure dominates the pe-BH procedure when both are valid (that is, any hypothesis rejected by pe-BH will also be rejected by ep-BH). To show this, we proceed as follows: the e-BH procedure with input $(e_1,\ldots,e_K)$ is equivalent to the p-BH procedure with input $(1/e_1,\ldots,1/e_K)$. Hence, the pe-BH procedure can be seen as applying the p-BH procedure to $(1/E^*_1,\ldots,1/E^{*}_K)$. 
If $h(p) >1/p$ for even a single $p\in (0,1)$, then $h$ is not a p/e calibrator. Indeed, for $P\sim \mathrm{U}(0,1)$, by decreasing monotonicity of $h$, we get
$\E\{h(P)\} \geq \E\{h(P) \id(P < p)\} >  \E\{\id(P < p)/p\} = \p(P < p)/p = 1$, which contradicts the fact that $h(P)$ is an e-value. 
Therefore, $p\mapsto 1/p$ is an upper bound for all p/e calibrators $h$, and this implies, for each $k\in \mathcal K$,
$$
\frac{1}{E^*_k} = \frac{1}{h(P_k) E_k} 
\ge \frac{P_k}{E_k} = P_k^* .
$$
Hence, the pe-BH procedure is dominated by the ep-BH procedure.

On the other hand, the ep-BH procedure requires some dependence assumption, such as positive regression dependence on a subset (Definition~\ref{definition:prds}).
By Theorem \ref{th:ep-BH}, 
if $\boldP$ is positive regression dependent on a subset, then the ep-BH procedure is valid, and it should be the better choice than the pe-BH procedure which is dominated.
However, if there is no dependence information of $\boldP$, then $\boldPstar$ is arbitrarily dependent, and
one may need to apply the p-BH procedure with the BY correction in \cite{BY01}.
In this case, the pe-BH and the ep-BH procedures do not dominate each other. In particular, one needs to compare the inputs
 $$
  \frac{1}{h(P_k) E_k}
\mbox{~~~and~~~} 
 \frac{ \ell_K P_k}{E_k}, \mbox{~~~where 
 }\ell_K:=\sum_{k=1}^K \frac 1 k \approx \log K.
  $$
This is analogous to the trade-off between the p-BH procedure with BY correction and the e-BH procedure, 
where one compares $h(P_k)$ and $(\ell_KP_k)^{-1}$ \citep[Section 6.5]{WR21}.

\subsection{A minimally adaptive e-BH procedure}
\label{sec:minimally_adaptive}

The discussion of the pe-BH procedure (and its comparison to ep-BH) raises the following question: can we use the i-pe/e combiner $\Pi_h$ within a procedure that controls the false discovery rate and is null-proportion adaptive (i.e., an analogous procedure to ep-Storey)?
The challenge here is that null proportion adaptive procedures analogous to Storey's are not known for the e-BH procedure (and consequently for the pe-BH procedure). 

In the remainder of this supplement we describe the first (minimally) adaptive procedure by proposing a tiny but uniform improvement of the e-BH procedure, inspired by \cite{SG17}.
We remark that, similarly to the situation of \cite{SG17}, this improvement is negligible for large values of $K$ and it may only be practically interesting for small $K$ such as $K\le 10$. 
We mainly focus on the case without boosting; see \cite{WR21} for  e-value boosting.

First, choose an e-merging function $F:[0,\infty]^K \to [0,\infty]$ in the sense of \cite{VW20}, i.e., $F$ satisfies that $F(E_1,\ldots,E_K)$ is an e-value for any e-values $E_1,\ldots,E_K$.  By Proposition 3.1 of \cite{VW20},  the arithmetic average 
$$
M:(e_1,\ldots,e_K)\mapsto \frac 1 K \sum_{k=1}^K e_k
$$
is the ``best"   symmetric e-merging function, in the sense that it is uniformly more powerful than any other symmetric e-merging functions.  We allow for a general choice of $F$ other than $M$ as it will be useful for the discussion later on boosted e-values.

With a chosen e-merging function $F$ and a level $\alpha\in (0,1)$, the improved e-BH procedure, denoted by $\cD^F(\alpha)$, is designed as follows. We 
  first test the global null $\bigcap_{k=1}^K H_k$ via the rejection condition $F(e_1,\ldots,e_K)\ge 1/\alpha$, which has a type-I error of at most $\alpha$, and if the global null is rejected,  we then apply the e-BH procedure at level $\alpha'=K\alpha/(K-1)$.
In other words, 
\begin{enumerate}
\item if $F(e_1,\ldots,e_K) < 1/\alpha$, then $\cD^F(\alpha)=\varnothing$;
\item  if $F(e_1,\ldots,e_K) \ge 1/\alpha$, then $\cD^F(\alpha)=\cD(\alpha')$ where $\alpha'=K\alpha/(K-1)$ and $\cD(\alpha')$ is the e-BH procedure at level $\alpha'$.
\end{enumerate}
The next proposition shows that by choosing $F=M$, 
the resulting improved BH procedure dominates the base BH procedure.

\begin{proposition}\label{prop:tiny}
The improved e-BH procedure $\cD^F(\alpha)$ applied to arbitrary e-values has false discovery rate at most $\alpha$.
In case $F=M$, $\cD^M(\alpha)$  dominates the  e-BH procedure $\cD(\alpha)$, that is, $\cD(\alpha)\subseteq \cD^M(\alpha)$. 
\end{proposition}

\begin{proof}
The first statement on false discovery rate  can be shown in a similar way as \cite{SG17}. 
Let
$A$ be the the event that $F(e_1,\dots,e_K)\ge 1/\alpha$, treated as random. 
If $K_0<K$, then, by using  Theorem 5.1 of \cite{WR21},  \begin{align*}  \E \left\{ \frac {F_{\cD^F(\alpha) }}{R_{\cD^F(\alpha)} } \right\}    &=  \E \left\{ \frac {F_{\cD(\alpha') }}{R_{\cD(\alpha')} } \id(A)\right\}   + \E \left[ \frac {F_{\varnothing}}{R_{\varnothing} } \{1-\id(A)\} \right]    \\&  = \E \left\{ \frac {F_{\cD(\alpha') }}{R_{\cD(\alpha')} } \id(A)\right\} \le \E \left\{ \frac {F_{\cD(\alpha') }}{R_{\cD(\alpha')} } \right\}  \le \frac{K_0}{K}\alpha' \le \alpha.
\end{align*}  
 If $K_0=K$, then the false discovery rate of  $\cD^F(\alpha)$ is at most the probability $\p(A)$ of rejecting the global null via $F(e_1,\dots,e_K)\ge 1/\alpha$. 
 In this case, $\p(A)\le \alpha$ by Markov's inequality and the fact that $F$ is an e-merging function.
Hence, in either case, the FDR of  $\cD^F(\alpha)$ is at most $\alpha$.

To show the second statement on dominance, let \begin{equation}\label{eq:e-simes}
S:(e_1,\dots,e_K)\mapsto \max_{k=1,\dots,K} \frac { k e_{[k]}} K.\end{equation} The function
$S$    is an e-merging function  and it is dominated by $M$ on 
 $[0,\infty]^K$; see Section 6 of \cite{VW20}.
Note that by definition, $S(e_1,\dots,e_K)<1/\alpha$ implies  $\cD(\alpha)=\varnothing$. 
Therefore, if $M(e_1,\dots,e_K)<1/\alpha$, then $\cD(\alpha)=\varnothing = \cD^M(\alpha)$.
Moreover, since $\alpha <\alpha'$, we always have $ \cD(\alpha) \subseteq \cD(\alpha')$.
Hence, $ \cD(\alpha) \subseteq \cD^M(\alpha)$.   
\end{proof}

Next, we briefly discuss the case of boosted e-values. The arithmetic average of boosted e-values is not necessarily a valid e-value, so one must be a bit more careful.
Nevertheless, it turns out that
we can use  the function $S$ in \eqref{eq:e-simes} on the boosted e-values. The new procedure can be described as the following steps.
\begin{enumerate}
\item Boost the raw e-values with level $\alpha$.
\item If $S(e'_1,\ldots,e'_K)<1/\alpha$ where $e'_1,\ldots,e'_K$ are the boosted e-values in step 1, then return $\varnothing$.
\item Else: boost  the raw  e-values with level $\alpha'=K\alpha/(K-1)$.
\item Return the discoveries by applying the base e-BH procedure to the boosted e-values in step 3.
\end{enumerate}

This new procedure dominates the e-BH procedure, and it has FDR at most $\alpha$. To show these two statements, it suffices to note that the probability of rejecting the global null test $S(e'_1,\ldots,e'_K)\ge 1/\alpha$ is at most $\alpha$ since the e-BH procedure has FDR at most $\alpha$ by Theorem 5.1 of \cite{WR21}; the rest of the proof is   similar to that of Proposition \ref{prop:tiny}.

\section{Using two samples  for the one-sided z-test} 
\label{sec:two_sample_stylized}

\subsection{Setup} 
We first describe the setup in more detail and more generality compared to our treatment in Section~\ref{subsec:stylized_two_sample_main}.
Suppose that we have two samples of  iid  data points, $X=(X_1,\dots,X_m)$  and $Y=(Y_1,\dots,Y_n)$, both from a distribution $\p$, where  $m\ge 0$ and $n\ge 1$. Here, if $m=0$ then $X$ has no data.
We would like to test $H_0: \p = P$ 
against $H_1: \p=Q $ where $P$ and $Q$ are distinct distributions.
For illustration, we will   focus on the simple case
$P=\mathrm{N}(0,1)$ and $Q=\mathrm{N}(  \delta,1)$ 
where $\delta >  0$ is known.

For an observation $x\in \R$, the likelihood ratio of $\mathrm{N}(\delta,1)$ over $\mathrm{N}(0,1)$ is  
$$
  L_\delta (x)
  :=
  \frac{\exp\{-(x-\delta)^2/2\}}{\exp(-x^2/2)}
  =
  \exp(\delta x - \delta^2/2). $$
  The likelihood ratio based on the sample $X $ is the e-value  $ E_{X}  $ for $Q$ given by
  $$
  E_{X}  :=\prod_{i=1}^m    L_{\delta }(X_i) =   \exp\left(\delta  \sum_{i=1}^m X_i - \frac{m \delta ^2}{2}\right),
  $$
  which has a log-normal distribution under $P$ with parameters $(\mu,\sigma^2)=(-m \delta ^2/2, m \delta ^2)$. Our convention is $E_{ X}=1$ if $m=0$.
In particular,  the  transformed log-likelihood statistic  $T_{ X}$ defined by  $$T_{ X}: =\frac{1}{\delta \sqrt m   }\left(\log   E_{ X} +\frac { m \delta ^2}2 \right)=\frac{1}{\sqrt m}\sum_{i=1}^m X_i$$ has a standard normal distribution  under $P$. 
 Based on the  statistic  $T_{ X}$, we can compute the likelihood ratio p-value 
 $$
 P_{ X} : =  1-\Phi(T_{ X}),
 $$
 where $\Phi$  is the standard normal distribution function. 
 Quantities like $P_{ Y}$  and $E_{ Y}$ are defined similarly with $(X_1,\dots,X_m)$ replaced by $(Y_1,\dots,Y_n)$.
  
 We consider three possible approaches to test the hypothesis using the two samples. 
 \begin{enumerate}[(a)]
  \item Combine two samples, that is, use $ Z=( X, Y)$ and then
  compute the p-value $P_{\rm LR}:= 1-\Phi (T_{ Z} )  $ based on the likelihood ratio of $ Z$,
  where $$T_{ Z}  :=\frac{1}{\sqrt {m+n }  }\left(  \sum_{i=1}^m X_i + \sum_{i=1}^nY_i \right).$$ 
 \item Compute an e-value $E_{ X}$ from $ X$ and a p-value $P_{ Y}$ from $ Y$, and use the $Q$-combiner to compute a p-value $P_{\rm E}:= Q(P_{ Y}, E_{ X}) =(P_{ Y}/ E_{ X}) \land 1$.\footnote{In what follows we slightly abuse notation and often ignore the truncation of $P_{ Y}/ E_{ X}$ to $1$. This does not make any difference with respect to the rejection of the null hypothesis.} 
 \item Compute a p-value $P_{ X}$ from $ X$ and another p-value $P_{ Y}$ from $ Y$, and use the Fisher statistic   $   -2\log ( P_{ Y} P_{ X})$.
The Fisher p-value is 
 $P_{\rm F}:=1- \chi_4(-2\log ( P_{ Y} P_{ X}))$,
 where $\chi_4$ is the the chi-square distribution with $4$ degrees of freedom.
  \end{enumerate} 

Among the three methods, our intuition is that $P_{\rm LR}$ should be the most powerful since it uses the full likelihood ratio of the sample. The other two methods, by combining two p-values or a p-value and an e-value, should lose some power.

\subsection{Pitman's asymptotic relative efficiency}

We   study Pitman's asymptotic relative efficiency (ARE; see \citealp[Section 14.3]{V98})  between the full likelihood ratio method and the P/E method.  
Note that if $n =0$ then both methods are   equivalent.
We  will consider that $ m,n\to \infty$ and $m=   \theta^2 n  $ with the signal ratio $\theta\ge 0$ fixed. 
(More precisely, one can use $m=\lfloor \theta^2 n \rfloor$ which does not make a difference to the asymptotic analysis.)
 
Fix two levels $\alpha$  and $\beta$ such that $1>\beta>\alpha>0$.
Let  $\p$ be the probability that generates the data from the alternative hypothesis,
and $Z'$ be a standard normal random variable under $\p$ independent of $X$.   
Define $N_{{\rm LR}}$ the number of sample points needed for a level-$\alpha$ test to reach power $\beta>\alpha$ under the alternative; that is, $N_{{\rm LR}}$ is the smallest number $n\in \N$ such that 
\begin{align}\label{eq:ARE1} 
\p( P_{\rm LR} \le \alpha) \ge \beta; \mbox{~~and equivalently, } \p\left\{Z'+\delta \sqrt{n\theta^2+n}  \ge  \Phi^{-1}(1-\alpha) \right\} \ge \beta.
\end{align}
Similarly, 
$N_{\rm E}$ is the smallest $n\in \N$ such that 
\begin{align}\label{eq:ARE2} 
\p( P_{\rm E} \le \alpha) \ge \beta; \mbox{~~and equivalently, } \p\left [ Z'+\delta \sqrt{n}    \ge  \Phi^{-1}\{(1-\alpha E_{ X})_+\} \right] \ge \beta.
\end{align}
Pitman's asymptotic relative efficiency between $P_{\rm LR}$ and $P_{\rm E}$ is defined as 
$$
\mathrm{ARE}_\theta(\alpha,\beta):=\lim_{\delta \downarrow 0} \frac{N_{{\rm LR}}}{N_{\rm E}},
$$
where we emphasize its reliance on $\theta$. 
The asymptotic relative efficiency intuitively means the ratio of the needed sample size for $P_{\rm LR}$ to that for $P_{\rm E}$ when the signal is very small.

\begin{proposition}\label{prop:ARE}
For $1>\beta >\alpha>0$ and $\theta \ge 0$, we have 
\begin{align}
\mathrm{ARE}_{\theta} (\alpha,\beta) =  \frac{z^2}{k^2}, \label{eq:ARE4}
\end{align}
where $k >0$ is the unique solution to 
\begin{align}
 \label{eq:y}\int_{-\infty}^{\infty} \Phi\left(\Phi^{-1}\left [ \{1-\alpha \exp(k^2 \theta^2 /2+k \theta w)\}_+\right ] -k \right)  \d  \Phi(w) &= 1-\beta,
\end{align} 
and $z>0$ is given by 
\begin{align}
 \label{eq:z}z =\frac{ \Phi^{-1}(1-\alpha  ) - \Phi^{-1}(1-\beta)  }{ \sqrt{1+\theta^2 } } .
 \end{align} 
 In particular, $\mathrm{ARE}_{0} (\alpha,\beta) = 1$.
Moreover, for fixed $\beta \in (0,1)$ and $\theta \ge 0$, we have 
\begin{align}
\lim_{\alpha \downarrow 0} \mathrm{ARE}_{\theta} (\alpha,\beta) = 1.\label{eq:ARE5}
\end{align}
\end{proposition}
\begin{proof}
First, let $\beta$ and $\alpha $ be fixed with $1>\beta >\alpha>0$.
It is easy to see that both $N_{\rm LR}$  and $N_{\rm E}$ tend to infinity as $\delta\downarrow 0$.
In this part of the proof, all convergence and asymptotic equivalence statements are with respect to $\delta\downarrow 0$.

The case of $N_{\rm LR}$ is easy to compute. By \eqref{eq:ARE1}, we have
\begin{align}
\label{eq:ARE3}
\frac{N_{\mathrm LR}\delta ^2 }{z^2 }\to 1 ~~~~\mbox{as $\delta \downarrow 0$,}
\end{align}
where $z$ is in \eqref{eq:z}.

Next, we analyze $N_{\rm E}$. Let $W_i=X_i- \delta$ for $i\in \N$ which is standard normally distributed under $Q$. 
Write $ E_{ X} (m)$ as $E_{ X}$ with sample size $m$. 
Note that 
$$
\log E_{ X}( m) =   \delta \sum_{i=1}^{m} W_i + \frac{m  \delta^2}{2},
$$
and therefore, 
\begin{align}
\label{eq:ARE-e}
\log E_{ X}(\theta ^2 n) -  \frac{n \theta ^2  \delta^2 }{2} =   \delta \sum_{i=1}^{ \theta ^2 n}W_i \sim \mathrm{N}(0, n \theta ^2  \delta^2 ).
\end{align} 
Let $    W $ be a standard normal random variable independent of $Z'$. 
We have 
$
\log E_{ X}(\theta ^2 n) 
$
is identically distributed as $ \delta \theta \sqrt{n} W + n \theta ^2  \delta^2 /2$.
 Using \eqref{eq:ARE-e} and   \eqref{eq:ARE2},  $N_{\rm E}  \sim
 k^2 /\delta^2  $ where $k>0 $ is such that
$$\p\left(Z'+ k  \ge  \Phi^{-1}\left[\{1-\alpha \exp(k ^2\theta^2 /2+k \theta W)\}_+\right] \ \right) =\beta.$$
By independence of $Z'$ and $W$,
$$\E\left \{ \Phi\left(\Phi^{-1}\left[ \{1-\alpha \exp(k^2\theta^2 /2+k\theta W)\}_+\right ] -k \right) \right \} = 1-\beta.$$
Equivalently, $k$ is such that 
\begin{align}\label{eq:k} \int_{-\infty}^{\infty} \Phi\left(\Phi^{-1}\left [\{1-\alpha \exp(k^2 \theta^2 /2+k \theta w)\}_+\right ] -k \right)  \frac{1}{\sqrt{2\pi}}e^{-w^2/2}\d w = 1-\beta.\end{align}
The equation \eqref{eq:k} is precisely \eqref{eq:y}, and it has a unique solution due to the strict monotonicity of the left-hand side of \eqref{eq:k} in $k$.
Using $N_{\rm E}  \sim
 k^2 /\delta^2  $ and $N_{\rm LR}  \sim
 z^2 /\delta^2  $, we get $\mathrm{ARE}_\theta (\alpha,\beta)\sim  z^2/k^2$.

Next, we prove \eqref{eq:ARE5}. In what follows, all convergence and asymptotic equivalence statements are with respect to $\alpha\downarrow 0$. We know that $k\to \infty$ since $k\geq z$ and $z \to \infty$ as $\alpha \downarrow 0$. 
For \eqref{eq:k} to hold, the  term 
$$f(\alpha):=\Phi^{-1}\left[\{1-\alpha \exp(k^2 \theta^2 /2+k \theta w)\}_+\right ] -k$$ 
 needs to be of the order $O(1)$ for some $w\in \R$.
We claim that for this to happen, we need
$$
  R_\alpha:= \frac{ k^2 (1+\theta^2)} {-2\log \alpha } \to 1, ~~~\mbox{~as $\alpha \downarrow 0$}.
$$
Note that 
\begin{align}
f(\alpha)= \Phi^{-1}\left[\{1-  \alpha^{1-R_\alpha   \theta^2 /(1+\theta^2)  } e^{k \theta w} \}_+\right ] -k.
\end{align}
Using the approximation $\Phi^{-1} (1-\epsilon)/ \sqrt {-2\log \epsilon} 
 \to 1$ as $\epsilon\downarrow 0$ (see Example 8.13 of \citealp{D08}), 
 we have 
$$ \Phi^{-1}\left[\{1-  \alpha^{1-R_\alpha   \theta^2 /(1+\theta^2)  } e^{k \theta w} \}_+\right ]
\sim \sqrt{-2 \left[ \{1-R_\alpha   \theta^2 /(1+\theta^2)  \} \log  \alpha  + {k \theta w}\right] },
 $$
 and by definition 
 $$
 k  = \sqrt{\frac{{R_\alpha (-2\log \alpha )}}{ {1+\theta^2}}}.
 $$
Putting the above two equations together, we get
\begin{align*}
 \frac{\Phi^{-1}\left[\{1-  \alpha^{1-R_\alpha   \theta^2 /(1+\theta^2)  } e^{k \theta w} \}_+\right ] }{k }
 &\sim \sqrt{ \frac{-2 \left[ \{1-R_\alpha   \theta^2 /(1+\theta^2)  \} \log  \alpha  \right ]  (1+\theta^2)}{ {R_\alpha (-2\log \alpha )} } -  \frac{{2 \theta w}}{k} }  
 \\& \sim \sqrt{ \frac{   1+\theta^2 -R_\alpha   \theta^2     }{ {R_\alpha } }  }  ,
 \end{align*}
 and the above term is asymptotically equivalent to $1$ if and only if $R_\alpha \to 1$.
 Since $k\to \infty$, if the above ratio is not $1$ then $f(\alpha)$ tends to $\infty$ or $-\infty$ for every $w$, violating \eqref{eq:k}. 
 From this, we  conclude that $R_\alpha \to 1$, and hence,
 $
 k^2\sim  -2\log \alpha /(1+\theta^2). 
 $
Using  the approximation $\Phi^{-1} (1-\epsilon)/ \sqrt {-2\log \epsilon} 
 \to 1$ again and \eqref{eq:z}, 
we have $z^2 \sim   -2\log \alpha /(1+\theta^2)$.
Therefore, we obtain \eqref{eq:ARE5}.
\end{proof}  

In Table \ref{tab:ARE} we report  numerical values for $\mathrm{ARE}_\theta (\alpha,\beta)$ for some choices of $\theta \in [0,1]$ and $\alpha, \beta$.
For instance, for $\theta=1$ and $(\alpha,\beta)=(0.05,0.9)$,  by using the P/E method compared to the full likelihood ratio, one at most loses $1$ data point in every $8$ data points. This remains true (as a conservative statement) for any value of $\theta \in [0,1]$, since $\mathrm{ARE}_\theta (\alpha,\beta)\ge   \mathrm{ARE}_1(\alpha,\beta)$ for $\theta \in [0,1]$.
If $X$ has less signal than $Y$, then the ARE is even closer to $1$. For instance, with the same $(\alpha,\beta)=(0.05,0.9)$,   $ \mathrm{ARE}_{0.5} (0.05,0.9)=0.956$, meaning that one loses $1$ data point in every 23 data points. 

\begin{table}[t!]
\begin{center}
\begin{tabular}{ c c  ccc  }
$(\alpha,\beta)$ & $(0.05,0.5)$ &$ (0.01,0.5)  $& $(0.05,0.9)$ &$ (0.01,0.9)$\\\hline
 $\mathrm{ARE}_1(\alpha,\beta)$ &     0.956 & 0.974 & 0.874 &0.914   \\ 
  $\mathrm{ARE}_{0.5}(\alpha,\beta)$ &     0.989 &  0.995 & 0.956 &0.970   \\  
   \end{tabular}
\end{center}
\caption{Asymptotic relative efficiency: We numerically compute the 
values of Pitman's asymptotic relative efficiency (following Proposition~\ref{prop:ARE}) between the full likelihood ratio method and the P/E method for different choices of the size $\alpha$, power $\beta$, and signal ratio $\theta$. For instance, with $\theta=0.5$, and $(\alpha,\beta)=(0.05,0.9)$,   $\mathrm{ARE}_{0.5} (0.05,0.9)=0.956$, meaning that one loses $1$ data point in every 23 data points. We observe that the asymptotic relative efficiency  is quite close to 1 and increases as $\alpha$ and $\theta$ decrease.} 
\label{tab:ARE}
\end{table}

\section{Extensions to the setting of Section~\ref{sec:unoormalized_data_driven_weighting}}
\label{sec:suppl_two_sample_ttests}

\subsection{Simultaneous two-sample t-tests}
\label{subsec:simultaneous_two_sample}
In place of~\eqref{eq:gaussian_replicates}
, consider a two-sample situation in which we observe (independent) $Y_{kj}, V_{kj}$ for $k=1,\dotsc,K$ and $j=1,\dotsc,n$, drawn as follows:\footnote{The assumption that the we observe the same number of observations, $n$, for each sample, is merely for notational convenience.}
\begin{equation}
\label{eq:gaussian_replicates_two_sample}
Y_{kj} \sim \mathrm{N}(\mu_{Y,k},\, \sigma_k^2),\; V_{kj} \sim \mathrm{N}(\mu_{V,k},\, \sigma_k^2),\;\; \mu_{Y,k}, \mu_{V,k} \in \mathbb R,\;\; \sigma_k > 0,
\end{equation}
We seek to test $H_k: \mu_{Y,k}= \mu_{V,k}$.  Let us compute the following,
$$
\hat{\mu}_{Y,k} := \frac{1}{n}\sum_{j=1}^n Y_{kj},\;\; \hat{\mu}_{V,k} := \frac{1}{n}\sum_{j=1}^n V_{kj},
$$
$$\hat{\sigma}_{Y,k}^2 := \frac{1}{n-1} \sum_{j=1}^n (Y_{kj} - \hat{\mu}_{Y,k})^2,\;\; \hat{\sigma}_{V,k}^2 := \frac{1}{n-1} \sum_{j=1}^n (V_{kj} - \hat{\mu}_{V,k})^2,
$$
$$
T_k := \frac{\sqrt{n}(\hat{\mu}_{Y,k} -\hat{\mu}_{V,k}) }{\sqrt{\hat{\sigma}_{Y,k}^2 + \hat{\sigma}_{V,k}^2}},\;\; P_k := 2\{1-F_{t, 2n-2}(|T_k|)\},
$$
where $F_{t, 2n-2}$ is the cumulative distribution function of a random variable following the t-distribution with $2n-2$ degrees of freedom. This is the p-value of the standard equal variance two-sample t-test.  Finally, let:
$$\hat{\mu}_k := \frac{1}{2}(\hat{\mu}_{Y,k} +\hat{\mu}_{V,k}),\;\; S_k^2 := \frac{1}{2n-1} \sum_{j=1}^n \{ (Y_{kj} - \hat{\mu}_k)^2 + (V_{kj} - \hat{\mu}_k)^2 \}.$$
Notice that $\hat{\mu}_k$ and $S_k^2$ are the sample mean and sample variance after pooling all observations $Y_{k,1},\dotsc,Y_{k,n}, V_{k,1},\dotsc,V_{k,n}$ and ignoring their group assignment. 

In analogy to~\eqref{eq:s_k_p_k_independence} in the manuscript, we can then show that under the null (i.e., when $\mu_{Y,k} = \mu_{V,k}$), 
then $P_k$ and $S_k^2$ are independent. For further context and references, see, e.g.,~\citet[Supplement S6.2.2]{ignatiadis2021covariate}.

The subsequent argumentation and methodological development could proceed analogously to the one-sample t-test problem that we study in the main text. We focus on the one-sample t-test instead of the two-sample t-test for the sake of simplicity, and notational compactness.

\subsection{Further extensions}

Analogous constructions that lead to an independence statement under the null of the form~\eqref{eq:s_k_p_k_independence} are available beyond the Gaussian models~\eqref{eq:gaussian_replicates} and~\eqref{eq:gaussian_replicates_two_sample}. For example, suppose that for the $k$-th hypothesis we
observe two samples, $Y_{k1},\dotsc,Y_{kn}$, and $V_{k1},\dotsc,V_{kn}$. We seek to conduct a nonparametric two-sample test.  If $\{Y_{k1},\dotsc,Y_{kn}, V_{k1},\dotsc,V_{kn}\}$ are assumed to be exchangeable for $k \in \mathcal{N}$, then any permutation-invariant statistic is independent of the Wilcoxon rank sum statistic. See \citet{bourgon2010independent} for further details. 

\section{More general weighting functions in the setting of Section~\ref{sec:unoormalized_data_driven_weighting}}
\label{sec:suppl_unnormalized}
Recall that our goal is to construct $\psi(\cdot)$ with the following two properties: it leads to more powerful e-values compared to $\psi(s^2)=s^2$, and second, a data-driven scaling analogous to~\eqref{eq:westfall_evalues} is practical and stable.

Our starting point is the likelihood ratio $L_k$  of $n S_k^2$ under the noncentral chi-square distribution with $n$ degrees of freedom and noncentrality parameter (ncp) $\lambda$ and under the (central)  chi-square distribution with $n$ degrees of freedom.\footnote{In fact, this likelihood ratio was the e-value we used in the simulation study of an earlier working paper of this work. The disadvantage of $\text{L}_k$ is that it is unclear how to scale it as in~\eqref{eq:westfall_evalues} when $\sigma_k^2>0$ is unknown.} We have the following expansion of $L_k$ in terms of powers of $S_k^2$:
$$L_k = \sum_{d=0}^{\infty} \frac{\exp(-\lambda/2) \lambda^d \Gamma(n/2)}{4^d d! \Gamma(n/2+d)} (n S_k^2)^d.$$
Our proposal is to fix $D \in \mathbb N$ and to truncate the above power series to the first $D+1$ terms, i.e.,
$$ \psi(S_k^2) \equiv L_k^D := \sum_{d=0}^{D}  \frac{\exp(-\lambda/2) \lambda^d \Gamma(n/2)}{4^d d! \Gamma(n/2+d)} (n S_k^2)^d.$$
In our implementation, we take $D=6$ and $\lambda = n$. Furthermore, let 
$$ \tilde{L}_k^D := \sum_{d=0}^{D} \frac{\exp(-\lambda/2) \lambda^d \Gamma\{(n-1)/2\}}{4^d d! \Gamma\{(n-1)/2+d\}} \{(n-1) \hat{\sigma}_k^2\}^d.$$
We may verify that for $k \in \mathcal{N}$, $\E(\tilde{L}_k^D) =\E(L_k^D)$. Hence this motivates the following choice of e-value rescaling analogous to~\eqref{eq:westfall_evalues}:
$$ E_k := K L_k^D \bigg / \sum_{\ell =1}^K \tilde{L}^D_k.$$
This is precisely~\eqref{eq:evalue_trunc}.

The conclusions of Theorem~\ref{th:data_driven_weights} hold verbatim after replacing the event $A_{\delta}$ by the event:
$$A_{\delta} := \left\{\sum_{k=1}^K \tilde{L}_k < (1-\delta) \sum_{k \in \mathcal{N}} \E(\tilde{L}_k)\right\}.$$

\section{E-values for replicated microarray data}
\label{appendix:evalues_for_microarrays}
We first provide a quick summary of the distributional assumptions and p-value constructions in~\citet{lonnstedt2002replicated} and \citet{smyth2004linear} and then derive analogous e-values. The starting point is that we seek to test $K$ hypotheses $H_k: \beta_k = 0$ wherein for the $k$-th hypothesis we have summarized our data as $\widehat{\beta}_k, S_k^2$, where
\begin{equation}
\label{eq:limma_distribution_assumptions}
\widehat{\beta}_k \mid \beta_k, \sigma_k^2 \sim \mathrm{N}(\beta_k,\, v_k \sigma_k^2),\;\;\; S_k^2 \mid \sigma_k^2 \sim \frac{\sigma_k^2}{\nu_k} \chi^2_{\nu_k}.
\end{equation}
Above, $v_k$ and $\nu_k$ are known fixed numbers and $\chi^2_{\nu_k}$ is the chi-square distribution with $\nu_k$ degrees of freedom. To be concrete, in case we conduct an equal variance two-sample t-test for each gene based on $n_k^C$ control samples and $n_k^T$ treated samples, then under standard normality assumptions we may take $v_k = (1/n_k^C + 1/n_k^T)$ and $\nu_k = n_k^T + n_k^C - 2$ in~\eqref{eq:limma_distribution_assumptions}.

To share information across genes,~\citet{lonnstedt2002replicated, smyth2004linear} further posit the following distributional assumption on the residual variances $\sigma_k^2$:
\begin{equation}
\label{eq:limma_shrinkage_assumption}
\frac{1}{\sigma_k^2} \sim \frac{1}{\nu_0 s_0^2} \chi^2_{\nu_0},
\end{equation}
where $s_0^2, \nu_0$ are fixed numbers that determine the location and concentration of the distribution of the $\sigma_k^2$. Under~\eqref{eq:limma_distribution_assumptions} and~\eqref{eq:limma_shrinkage_assumption}, it also follows that,
\begin{equation}
\label{eq:moderated_t}
\widetilde{T}_k \mid \beta_k =0 \, \sim\, t_{\nu_0 + \nu_k},\; \text{ where }\,    \widetilde{T}_k := \frac{\widehat{\beta}_k}{ \tilde{S}_k\sqrt{v_k}}, \;\; \tilde{S}_k^2 := \frac{\nu_0 s_0^2 + \nu_k S_k^2}{\nu_0 + \nu_k},
\end{equation}
where $t_{\nu_0 + \nu_k}$ is the t-distribution with $\nu_0 + \nu_k$ degrees of freedom with cumulative distribution function $F_{t,\nu_0 + \nu_k}$. Hence, $P_k = 2\{1-F_{t,\nu_0 + \nu_k}(|\widetilde{T}_k|)\}$ is a p-value for the null hypothesis $H_k: \beta_k =0$. The upshot of positing~\eqref{eq:limma_shrinkage_assumption} is that
we may studentize $\widehat{\beta}_k$ with sample variances that are shrunk toward $s_0^2$, and increase the degrees of freedom of the t-statistic from $\nu_k$ to $\nu_0 + \nu_k$.\footnote{One may wonder if the additional assumption~\eqref{eq:limma_shrinkage_assumption} is justified. For microarray data and RNA-Seq data analyzed via \texttt{limma}~\citep{ritchie2015limma}, \eqref{eq:limma_shrinkage_assumption} often provides an adequate fit with respect to downstream inferences~\citep{lu2016variance}. \citet{lu2016variance} and \citet{ignatiadis2023empirical} replace~\eqref{eq:limma_shrinkage_assumption} by a nonparametric generalization.} Model~\eqref{eq:limma_shrinkage_assumption} and the computation of the p-values entails knowledge of the two unknown parameters $s_0^2, \nu_0$. The predominant approach in practice, as advocated by~\citet{smyth2004linear}, is to plug-in parametric empirical Bayes estimates $\widehat{\nu}_0$ and $\widehat{s}_0^2$ of these parameters (based on the data for all genes,  $k=1,\dotsc,K$).

Taking inspiration from the above p-value construction, we propose the following e-value.

\begin{proposition}
\label{prop:limma_evalue}
Suppose~\eqref{eq:limma_distribution_assumptions} and~\eqref{eq:limma_shrinkage_assumption} hold. Then, for any $\gamma>0$,
\begin{equation}
\label{eq:microarray_evalue}
E_k := \frac{1}{\sqrt{\gamma_k + 1}}\left\{ 1 - \frac{\gamma_k \widetilde{T}_k^2}{(1 + \gamma_k)(\nu_k + \nu_0 + \widetilde{T}_k^2)}\right\}^{- \frac{\nu_0 + \nu_k + 1}{2}}, \; \gamma_k = \gamma / v_k,
\end{equation}
is an e-value for $H_k: \beta_k=0$, in particular, $\E(E_k \mid \beta_k=0) = 1$.
\end{proposition}

\begin{proof}
Let $p_{k,0}(\cdot)$ be the (marginal) density of the moderated t-statistic~$\widetilde{T}_k$~\eqref{eq:moderated_t} under  \eqref{eq:limma_distribution_assumptions}, \eqref{eq:limma_shrinkage_assumption}, and $\beta_k=0$. For $\gamma > 0$, let $p_{k,\gamma}(\cdot)$ be the (marginal) density of the moderated t-statistic~$\widetilde{T}_k$ when $\beta_k \mid \sigma_k^2 \sim \mathrm{N}(0, \gamma \sigma_k^2)$ and \eqref{eq:limma_distribution_assumptions}, \eqref{eq:limma_shrinkage_assumption} hold. Then $E_k$ in~\eqref{eq:microarray_evalue} is equal to the likelihood ratio $p_{k, \gamma}(\widetilde{T}_k)/p_{k, 0}(\widetilde{T}_k)$. Hence:
$$\E(E_k \mid \beta_k=0) = \int \{ p_{k, \gamma}(t)/p_{k, 0}(t)\} p_{k, 0}(t) dt = \int p_{k, \gamma}(t) dt = 1.$$
\end{proof}

The e-value construction above requires a choice of a tuning parameter $\gamma >0$. The proof above hints at a way of choosing $\gamma$ in a data-driven way. We make the additional working model assumption:
\begin{equation}
\label{eq:half_half_model}
\beta_k \mid \sigma_k^2 \sim \tilde{\pi}_0 \delta_0 + (1-\tilde{\pi}_0) \mathrm{N}(0, \gamma \sigma_k^2),\;\; \tilde{\pi}_0=1/2,
\end{equation}
where $\delta_0$ is a point mass at $0$.
We then estimate $\gamma$ by empirical Bayes as described in~\citet[Section 6.3]{smyth2004linear} by positing that~\eqref{eq:half_half_model} holds for all genes $k=1\dotsc,K$ in addition to~\eqref{eq:limma_distribution_assumptions} and~\eqref{eq:limma_shrinkage_assumption}. Analogously to the computation of p-values in \texttt{limma}~\citep{smyth2004linear}, we ignore uncertainty introduced due to the estimation of $\gamma$. 

The choice $\tilde{\pi}_0=1/2$ in~\eqref{eq:half_half_model} is a conservative choice. If we were to further increase the proportion assigned to the null component $(\tilde{\pi}_0)$, then the estimated $\gamma$ would typically be larger, and this would lead to more extreme e-values. Instead, we make the safe choice $\tilde{\pi}_0=1/2$ in anticipation of the downstream task of combining e-values with p-values. In particular, we emphasize, that inferences will be valid even if the true null proportion $\pi_0$ is different than our posited $\tilde{\pi}_0$.

The e-value $E_k$ has the following elegant interpretation for $\tilde{\pi}_0=1/2$: it is equal to the posterior odds statistic proposed by \citet[Equation 3]{lonnstedt2002replicated}. The posterior odds statistic relies on the validity of~\eqref{eq:half_half_model}, and this led \citet{lonnstedt2002replicated} to write that ``we cannot rely on any standard cutoff value [...] for the selection of differentially expressed genes.'' However, since the posterior odds statistic is an e-value, we no longer need to rely on~\eqref{eq:half_half_model} (it is merely a working assumption), and we can rigorously proceed with the multiple testing task.

\section{Additional simulation figures}
\label{sec:add_simulations_figures}
\begin{figure}[H]
    \centering
    \includegraphics[width=\linewidth]{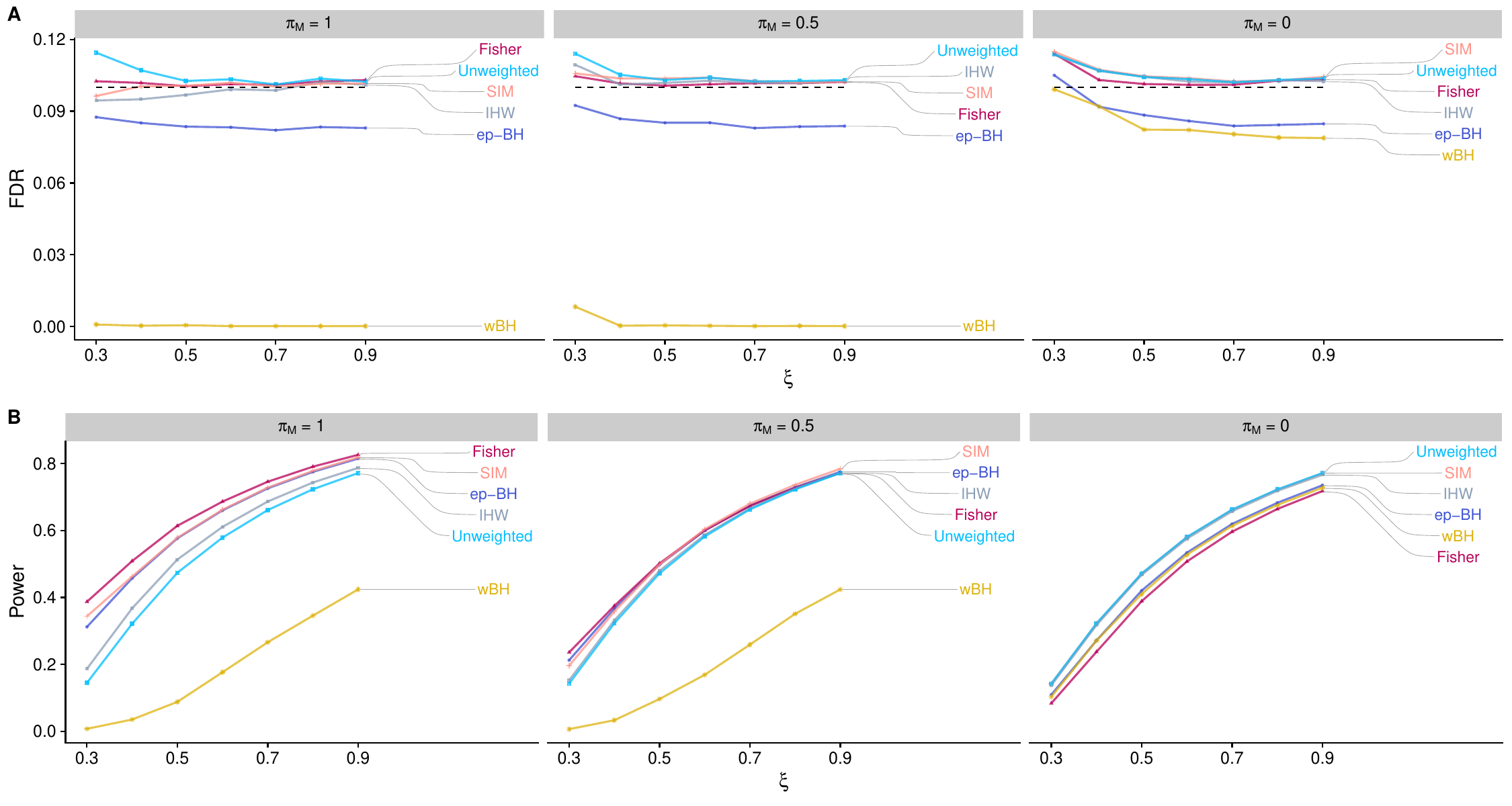}
    \caption{RNA-Seq and microarray meta-analysis simulation with null-proportion adaptive methods: This figure is analogous to Fig.~\ref{fig:rnaseq_sim_main} with the difference that we compare null-proportion adaptive variants of the same methods. We plot A) the false discovery rate (FDR) and B) power against the effect size parameter $\xi$ and against the informativeness of the microarray data (parameter $\pi_M$ in the facets).  We note that in this case there is slight exceedance of FDR control for several methods (including unweighted p-BH) at small values of $\xi$ (which is also slightly visible---but less so---in Fig.~\ref{fig:rnaseq_sim_main}). The reason may be that \texttt{DESeq2}~\citep{love2014moderated} p-values are computed based on asymptotic approximations, and so may not be exactly super-uniform in finite samples. The remaining takeaways are similar to those of Fig.~\ref{fig:rnaseq_sim_main}: Fisher Storey-BH has the most power when microarray data are fully informative ($\pi_M=1$), but has the least power when the microarray data are fully uninformative $(\pi_M=0)$.
    }
    \label{fig:rnaseq_sim_suppl}
\end{figure}

\end{document}